\documentclass[11pt,a4paper]{article}
\usepackage{amssymb,amsmath,amsfonts,amsthm,amstext,amscd,array}
\usepackage{mathrsfs}
\usepackage{hyperref}
\usepackage{pdfsync}
\usepackage{bbm}
\usepackage{bm}
\usepackage[arrow,matrix,curve]{xy}
\usepackage{bbding}
\usepackage{wasysym}

\DeclareMathOperator{\ad}{ad}

\def\nn{\nonumber}

\newcommand{\eps}{\varepsilon}

\newcommand{\mbf}[1]{{\boldsymbol {#1} }}
\def\ii{{\,{\rm i}\,}}
\def\dd{{\rm d}}

\def\Id{{\rm id}}

\def\Diff{{\sf Diff}}

\def\mfh{{\mathfrak h}}

\def\Lie{{\mathcal L}}

\newcommand{\eq}{\begin{equation}}
\newcommand{\eqend}{\end{equation}}
\newcommand{\eqa}{\begin{eqnarray}}
\newcommand{\nonueqa}{\begin{eqnarray*}}
\newcommand{\eqaend}{\end{eqnarray}}
\newcommand{\nonueqaend}{\end{eqnarray*}}

\newcommand{\bma}[1]{\begin{array}{#1}}
\newcommand{\ema}{\end{array}}
\newcommand{\bc}{\begin{center}}
\newcommand{\ec}{\end{center}}

\newcommand{\vect}{{\sf Vect}}

\renewcommand{\thefootnote}{\fnsymbol{footnote}}
\newcommand{\newsection}{\setcounter{equation}{0}\section}

\newcommand{\complex}{{\mathbb C}} %% complex numbers
\newcommand{\real}{{\mathbb R}} %% real numbers

\def\alg{{\mathcal A}}

\def\Lcal{{\mathcal L}}
\def\Mcal{{\mathcal M}}

\def\ot{{\, \otimes\, }}

\newif\ifold             \oldtrue

\def\nn{\nonumber}

\def\e{{\,\rm e}\,}

\def\be{\begin{equation}}
\def\ee{\end{equation}}
\def\bea{\begin{eqnarray}}
\def\eea{\end{eqnarray}}
\def\bd{\begin{displaymath}}
\def\ed{\end{displaymath}}

\newcommand{\beq}{\begin{eqnarray}}
\newcommand{\eeq}{\end{eqnarray}}

\makeatletter
\newdimen\normalarrayskip              % skip between lines
\newdimen\minarrayskip                 % minimal skip between lines
\normalarrayskip\baselineskip
\minarrayskip\jot
\newif\ifold             \oldtrue            
\def\arraymode{\ifold\relax\else\displaystyle\fi} % mode of array entries
\def\@arrayskip{\ifold\baselineskip\z@\lineskip\z@
     \else
     \baselineskip\minarrayskip\lineskip2\minarrayskip\fi}
\def\@arrayclassz{\ifcase \@lastchclass \@acolampacol \or
\@ampacol \or \or \or \@addamp \or
   \@acolampacol \or \@firstampfalse \@acol \fi
\edef\@preamble{\@preamble
  \ifcase \@chnum
     \hfil$\relax\arraymode\@sharp$\hfil
     \or $\relax\arraymode\@sharp$\hfil
     \or \hfil$\relax\arraymode\@sharp$\fi}}
\def\@array[#1]#2{\setbox\@arstrutbox=\hbox{\vrule
     height\arraystretch \ht\strutbox
     depth\arraystretch \dp\strutbox
     width\z@}\@mkpream{#2}\edef\@preamble{\halign \noexpand\@halignto
\bgroup \tabskip\z@ \@arstrut \@preamble \tabskip\z@ \cr}%
\let\@startpbox\@@startpbox \let\@endpbox\@@endpbox
  \if #1t\vtop \else \if#1b\vbox \else \vcenter \fi\fi
  \bgroup \let\par\relax
  \let\@sharp##\let\protect\relax
  \@arrayskip\@preamble}
\makeatother

\newcommand{\p}{\partial}

\def\tP{{\tilde P}}

\def\As{{\APLstar}}

\def\uf{{\underline f}}
\def\ug{{\underline g}}
\def\uh{{\underline h}}
\def\uk{{\underline k}}
\def\uu{{\underline u}}
\def\uv{{\underline v}}
\def\uz{{\underline z}}

\def\FF{{\cal F}}
\def\al{\alpha}
\def\be{\beta}

\newtheorem{proposition}[equation]{Proposition}

\newtheorem{corollary}[equation]{Corollary}

\theoremstyle{definition}

\def\pphim{{\bar\phi\bar\phi}}
\allowdisplaybreaks

\usepackage{hyperref}
\hypersetup{colorlinks,%
citecolor=black,%
filecolor=black,%
linkcolor=black,%
urlcolor=black}

\begin{document}

\begin{titlepage}

\begin{flushright}

\baselineskip=12pt

EMPG--15--06

\end{flushright}

\begin{center}

\vspace{1cm}

\baselineskip=24pt

{\Large\bf Triproducts, nonassociative star products \\  and geometry of $\mbf
  R$-flux string compactifications}\footnote{Contribution to the Proceedings of ``Conceptual and Technical
Challenges for Quantum Gravity'', September 8--12, 2014, Sapienza
University of Rome, Italy; to be published in {\sl Journal of Physics
  Conference Series}, eds. G.~Amelino-Camelia, P.~Martinetti and J.-C.~Wallet.}

\baselineskip=14pt

\vspace{1cm}

{\bf Paolo Aschieri}${}^{1}$ \ and \ {\bf Richard
  J. Szabo}${}^{2}$
\\[5mm]
\noindent  ${}^1$ {\it Dipartimento di Scienze e Innovazione
  Tecnologica}\\ and {\it INFN Torino, Gruppo collegato di
  Alessandria}\\ {\it Universit\`a del Piemonte Orientale}\\
{\it Viale T. Michel 11, 15121 Alessandria, Italy}\\
Email: \ {\tt
    aschieri@unipmn.it}
\\[3mm]
\noindent  ${}^2$ {\it Department of Mathematics, Heriot-Watt University\\ Colin Maclaurin Building,
  Riccarton, Edinburgh EH14 4AS, U.K.}\\ and {\it Maxwell Institute for
Mathematical Sciences, Edinburgh, U.K.} \\ and {\it The Higgs Centre
for Theoretical Physics, Edinburgh, U.K.}\\
Email: \ {\tt R.J.Szabo@hw.ac.uk}
\\[30mm]

\end{center}

\begin{abstract}

\baselineskip=12pt

We elucidate relations between different approaches to describing the nonassociative deformations of
geometry that arise in non-geometric string theory. We demonstrate how
to derive configuration space triproducts exactly from nonassociative phase space
star products and extend the relationship in various directions. By
foliating phase space with leaves of constant momentum we obtain
families of Moyal--Weyl type deformations of triproducts, and we
generalize them to new triproducts of differential forms and of tensor fields. We prove
that nonassociativity disappears on-shell in all instances. We also
extend our considerations to the differential geometry of
nonassociative phase space, and study the induced deformations of
configuration space diffeomorphisms. We further develop general
prescriptions for deforming configuration space geometry from the
nonassociative geometry of phase space, thus paving the way  to
a nonassociative theory of gravity in non-geometric flux
compactifications of string theory.

\end{abstract}

\end{titlepage}
\setcounter{page}{2}

\newpage

{\baselineskip=12pt
\tableofcontents
}

\bigskip

\renewcommand{\thefootnote}{\arabic{footnote}}
\setcounter{footnote}{0}

\newsection{Introduction and summary\label{sec:intro}}

Noncommutative geometry has long
been believed to provide a framework for understanding
the generalizations of classical spacetime geometry required to describe Planck scale
quantum geometry, and ultimately quantum gravity. Yet its precise
relation with other approaches to quantum gravity such as string theory
have remained somewhat elusive. The surge of interest
in noncommutative geometry from string theory originally came about in the
\emph{open} string sector, where the massless bosonic modes contain
gauge and scalar fields: It was observed that D-brane worldvolumes
acquire a noncommutative deformation in the background of a non-zero
$B$-field~\cite{Chu1999,Schomerus1999,Seiberg1999}, which is moreover nonassociative when the flux $H=\dd
B$ is
non-vanishing~\cite{Cornalba2001}; the low-energy worldvolume effective
field theory is then
described by a noncommutative gauge theory. However,
this realization does not really shed light on connections with
the \emph{closed} string sector
whose massless bosonic modes contain gravitational degrees of freedom
such as the metric $g$, the $B$-field, and the
dilaton $\phi$. Vanishing of the beta-functions for these fields (at
one-loop) as required by conformal invariance of the worldsheet field theory
yields equations of motion which can be derived from the standard bosonic
closed string low-energy effective action
\beq
S= \int_M \, R(g)-\frac1{12} \,
\e^{-\phi/3}\, H\wedge *_g \, H-\frac16 \, \dd \phi
\wedge*_g\, \dd\phi \ .
\eeq
It can be hoped that an equivalent
noncommutative/nonassociative version of this effective field theory provides a
suitable target space framework in which to address problems related
to quantum gravity.

A precise connection between noncommutative geometry and
the closed string sector has been found recently through non-geometric
flux compactifications~\cite{Hull2004,Shelton2005}. The prototypical example involves a three-torus
with constant three-form $H$-flux; T-dualizing along its cycles gives rise to
constant geometric and non-geometric fluxes which can be depicted schematically
through the T-duality chain
\beq
H_{abc} \ \xrightarrow{ \ T_a \ } \ f^a{}_{bc} \ \xrightarrow{ \ T_b \
} \ Q^{ab}{}_c \ \xrightarrow{ \ T_c \ } \ R^{abc} \ .
\label{eq:Tchain}\eeq
While the $H$-flux and metric $f$-flux backgrounds can be described
globally as Riemannian manifolds, the $Q$-flux background involves
T-duality transformations as transition functions between local trivializations of its
tangent bundle and gives rise to a T-fold, while the $R$-flux background is a
purely non-geometric string vacuum as its metric and
$B$-field are not even locally defined because they depend on the winding
coordinates of the dual space. These non-geometric flux
compactifications have been recently purported to
have a global description in terms of noncommutative and
nonassociative
structures~\cite{Lust2010,Blumenhagen2010,Blumenhagen2011,Mylonas2012}
(see~\cite{Mylonas2014} for a recent review and further references).

Just as in the case of open strings, there are two ways in which one
can see the appearence of nonassociativity of target space coordinates
in the $R$-flux background: Either through canonical analysis of
the closed string coordinates or by
careful examination of worldsheet scattering amplitudes. In the formulation of~\cite{Lust2010}, nonassociativity of the
coordinates of
\emph{configuration space} arises
as a failure of the Jacobi identity of the bracket of canonical
variables on \emph{phase space}. This model is described in
Section~\ref{sec:NAdef}, and it defines a twisted Poisson
structure on phase space which can be quantized using
Kontsevich formality to obtain an explicit nonassociative star product of functions for
deformation quantization of phase
space~\cite{Mylonas2012}. The nonassociative star product on phase space defines
nonassociative tori in closed string theory. If we regard it as a
substitute for canonical quantization as suggested
by~\cite{Bakas2013}, then it may be used to formulate a consistent
nonassociative version of quantum mechanics~\cite{Mylonas2013}. On the other hand, from a worldsheet
perspective an exact analysis of conformal field theory correlation functions can be carried out to linear order in the $H$-flux, wherein
the backreaction due to the curvature of spacetime can be neglected. This was
performed by~\cite{Blumenhagen2011} where off-shell traces of
nonassociativity were found and packaged into ``triproducts'' of
fields directly on configuration space at linear order in $R$; at this order the elementary triproduct
of three fields was shown by~\cite{Mylonas2012} to be reproduced by
the associator of the phase space star
product. In~\cite{Blumenhagen2011} an all orders exponential formula for
arbitrary $n$-triproducts involving $n$ fields was moreover
conjectured, however its verification through a conformal field theory
analysis on flat space is not possible.

In this contribution we fill a technical gap in these two
approaches to the nonassociative geometry of $R$-flux
compactifications by explicitly \emph{deriving} the conjectural all
orders configuration space
$n$-triproducts of~\cite{Blumenhagen2011} from the phase space star
products of~\cite{Mylonas2012}. As their origins are quite different, this provides a non-trivial
confirmation of these triproducts and moreover of the validity of the phase
space formulation as the fundamental theory of closed strings
propagating in the non-geometric $R$-flux compactification, which is
the perspective advocated by~\cite{Mylonas2012} and pursued in the
present paper wherein phase space appears as the effective closed
string target space. Physical applications of this
perspective to scalar field theories can be found
in~\cite{Mylonas2014a}. 

For later reference, let us briefly review the phase space
approach. The starting point of~\cite{Mylonas2012} is the Courant
sigma-model based on the standard Courant algebroid $C = TM\oplus
T^*M$ over the target space $M$ with 
  natural frame $(\psi_I)=(\partial_i,\dd x^i)$; the structure maps
  consist of the fibre metric
  $h_{IJ}$ with $h_{ij}= \langle\partial_i,\dd x^j\rangle= \delta_i{}^j$, the anchor matrix $\rho(\psi_I)=
  P_I{}^i(x)\,\partial_i$,\footnote{Implicit summation over repeated upper and lower indices
  is understood throughout.}
 and the three-form
  $T_{IJK}(x)=[\psi_I,\psi_J,\psi_K]_C$. This data can be used to
  construct a topological field theory on a three-dimensional M2-brane
  worldvolume $\Sigma_3$ with action given by
\begin{eqnarray}
S_T = \int_{\Sigma_3}\, \phi_i\wedge \dd X^i+\frac12\,
h_{IJ}\, \alpha^I\wedge \dd\alpha^J-P_I{}^i(X)\, \phi_i\wedge
\alpha^I +
\frac16\, T_{IJK}(X)\, \alpha^I\wedge\alpha^J \wedge\alpha^K \ ,
\end{eqnarray}
where $X:{\Sigma_3} \to M$ is the M2-brane embedding, and $\alpha \in
\Omega^1(\Sigma_3,X^*C)$ and $\phi \in 
\Omega^2({\Sigma_3},X^*T^*M)$ are auxilliary fields. By taking $T$ to
be any of the fluxes in (\ref{eq:Tchain}), one can look at the
effective dynamics of the membrane fields; putting in all of the
fluxes simultaneously and independently of one another (e.g., by
replacing $M$ with its double manifold as in double field theory)
could then give a precise formulation of the formal T-duality chain
(\ref{eq:Tchain}) in terms of (quantum) gauge symmetries of the Courant
sigma-model. In the case of a purely geometric $H$-flux background it
is shown in~\cite{Mylonas2012} that the M2-brane model reduces to the
standard $H$-twisted Poisson sigma-model with target space $M$ on the
boundary $\Sigma_2=\partial\Sigma_3$, which describes closed string
propagation in the configuration space $M$ with a non-constant $B$-field. On the other hand, in the
case of a non-geometric constant $R$-flux the M2-brane model reduces
instead to a generalized Poisson sigma-model on $\Sigma_2$ whose action
can be expressed in linearized form by using auxilliary one-form fields $\eta_I$ as
\beq
S_R= \int_{\Sigma_2}\, \eta_I\wedge\dd X^I+\frac12\,
\Theta^{IJ}(X) \, \eta_I\wedge\eta_J \ ,
\eeq
where  $X= (X^i,P_i) :\Sigma_2\to T^*M$ now embed
closed strings into an effective target space which coincides with the
phase space of $M$, and
\beq
\Theta=\begin{pmatrix} R^{ijk}\, p_k & \delta^i{}_j \\
  -\delta_i{}^j & 0 \end{pmatrix}
\eeq
is a twisted Poisson bivector which describes a noncommutative and
nonassociative phase space (see Section~\ref{sec:NAdef}). This point of
view that closed strings should propagate in phase space rather than
configuration space is also
pursued in a more general context by~\cite{Freidel2013} wherein it is argued that the fundamental
symmetry of string theory contains phase space diffeomorphisms in
order to accomodate T-duality and other stringy features. 

In Section~\ref{sec:star} we shall in fact derive a generalization to
families of Moyal--Weyl type deformations of the triproducts which are
induced from a foliation of phase space
by leaves of constant momentum and describe their
physical significance; the leaf
of zero
momentum yields the triproducts
of~\cite{Blumenhagen2011}. We further give an explicit proof that
the integrated versions of these deformed triproducts coincide with
those of the associative Moyal--Weyl star products, thus reproducing
the physical expectations that nonassociativity is not visible in
on-shell conformal field theory correlation
functions. This is a remarkable cancellation that happens due
to specific orderings of  bracketing phase space star products. 
In Section~\ref{sec:cochain} we generalize these results to new (deformed) triproducts of
arbitrary differential forms via nonassociative deformations of phase space exterior products
following the cochain twist deformation formalism proposed by~\cite{Mylonas2013}  for the study of non-geometric fluxes.
These algebraic techniques, originally developed to study quasi-Hopf
algebra symmetries \cite{Drinfeld}, are further extended to study
the differential geometry of nonassociative phase space, i.e., the
deformed Lie algebra of infinitesimal diffeomorphisms and their Lie derivative action
on exterior forms and tensor fields.
This geometry is then induced on each constant momentum leaf leading
in particular to a deformed Lie algebra of infinitesimal
diffeomorphisms on configuration space.

While the results of the present paper are mostly technical, 
they offer further insight into the nature of
nonassociative geometry in non-geometric string theory in at least three
ways. Firstly, we give a fairly general prescription for 
inducing configuration space triproducts from nonassociative star
products on phase space and clarify in what sense 
these products obey the physical
expectations of cyclicity of on-shell string
scattering amplitudes, though a complete geometric framework for
choosing appropriate section slices remains to be worked out. 
The different leaves 
correspond to a
multitude of choices of Bopp shifts in phase space and are reminescent
of the section conditions in double field theory (the weak and strong
constraints). Secondly, our considerations based on cochain twist
theory are general enough that
they may be applicable beyond the flat space limit to the
\emph{curved} space triproducts conjectured by~\cite{Blumenhagen2013}
for non-constant fluxes, and hence they may provide a precise link between
the phase space formalism and the framework of double field
theory; such a connection is explored by~\cite{Freidel2013} in a
related but more general context. Thirdly, our derivation and considerations of new triproducts
involving arbitrary differential forms and infinitesimal diffeomorphisms on configuration space is a
first step to understanding the construction of nonassociative
deformations of gravity and their physical relevance in on-shell
string theory.

\newsection{Nonassociative geometry of $R$-space\label{sec:NAdef}}

\subsection{Phase space formulation}

We shall begin by describing the phase space model for the non-geometric $R$-flux background and discuss some ways in which it may be interpreted as a deformation of the geometry of configuration space. Let $M=\real^d$ be the decompactification limit of a $d$-torus, endowed with
a constant three-form $H$-flux. As reviewed in
Section~\ref{sec:intro}, via T-duality this space is mapped to a
non-geometric string background with a constant trivector
$R$-flux $R=\frac1{3!}\, R^{ijk}\, \partial_i\wedge\partial_j\wedge \partial_k$,
where $\partial_i=\frac\partial{\partial x^i}$ in local coordinates
$x=(x^i)\in \real^d$. Explicit string and
conformal field theory computations show that the string geometry
acquires a noncommutative and nonassociative deformation for closed
strings which wind and propagate in the non-geometric
background~\cite{Blumenhagen2010,Lust2010,Blumenhagen2011,Condeescu2012,Andriot2013}. An
explicit realization of such a nonassociative deformation of the
spacetime geometry is provided by the phase space description of the parabolic
$R$-flux model on the cotangent bundle $\Mcal=T^*M=M\times
(\real^d)^*$ with local coordinates $(x,p)$, where
$p=(p_i)\in(\real^d)^*$. In this setting, the deformation is described
by a commutator algebra for the local phase space coordinates given
by
\beq
[x^i,x^j]=\frac{\ii \ell_s^4}{3\hbar}\, R^{ijk} \,p_k \ 
,\qquad [x^i,p_j]= \ii \hbar \, \delta^i{}_j \qquad\mbox{and} \qquad
[p_i,p_j]=0 \ , \label{luestalg}
\eeq
which has a non-trivial Jacobiator
\beq
[x^i,x^j,x^k]= \ell_s^4\, R^{ijk} \ .
\label{eq:3bracket}\eeq
In the point particle limit $\ell_s=0$ this is just the usual
Lie algebra of ordinary quantum phase space.

As anticipated by~\cite{Blumenhagen2010,Lust2010,Blumenhagen2013}, and
proven in~\cite{Mylonas2013} directly from the phase space model, the $R$-flux background does not permit
the notion of a point due to a minimal volume which enters an
uncertainty relation for the position coordinates given by 
\beq
\Delta
x^i\, \Delta x^j\, \Delta x^k\geq \mbox{$\frac12$} \, \ell_s^4\, R^{ijk} \ .
\eeq
This lack of a
notion of a point illustrates why the $R$-flux compactification is not
even locally geometric; this is evident in the phase space model
which, as a result of T-duality, requires both position coordinates
and momenta in any local description. Hence it is not clear how to
formulate a gravity theory, or any other field theory, on this
nonassociative space. An approach based on fundamental loop space
variables, rather than functions on $M$, was pursued in~\cite{Saemann2011,Saemann2013}; the usage of
such variables is also mentioned in~\cite{Blumenhagen2013}. This approach
reflects the ``non-locality'' of the non-geometric $R$-space. It
is natural in the lift of Type~IIA string theory
to M-theory in which the (closed string) boundary of an open M2-brane ending on an
M5-brane in a constant $C$-field background gives rise to a noncommutative loop space algebra on the
M5-brane worldvolume~\cite{Bergshoeff2000,Kawamoto2000}, which corresponds to the
noncommutativity and nonassociativity felt by a fundamental closed
string in a constant $H$-flux; this perspective is utilized in the
formulation by~\cite{Mylonas2012} of closed string propagation in the
non-geometric $R$-flux background using M2-brane degrees of freedom,
as reviewed in Section~\ref{sec:intro}, and it connects open and
closed string noncommutative geometry.

Another set of fundamental variables 
is obtained by considering the algebra $\Diff(M)$ of (formal)
differential operators on $M$ with typical elements of the form
\beq
\underline{f}(x) = f(x) +\sum_{k=1}^\infty\, f^{i_1\cdots
  i_k}(x)\, \partial_{i_1}\cdots \partial_{i_k} \ ,
\eeq
where $f,f^{i_1\cdots i_k}\in A:=C^\infty(M)$ and 
$\underline{f}(x)g(x) = f(x)\, g(x) +\sum_{k}\, f^{i_1\cdots
  i_k}(x)\, \partial_{i_1}\cdots \partial_{i_k} g(x)$  for all $g\in A$.
 From the last two commutation relations in (\ref{luestalg}) we can identify
$p_i=-\ii\hbar \, \partial_i$ and the nonassociativity relations
as relations involving differential operators
\beq
[x^i,x^j]=\mbox{$\frac13$} \, \ell_s^4 \, R^{ijk} \,\partial_k
\ .
\eeq
 Such an interpretation was 
advocated by~\cite{Ho2000,Ho2001} in the (associative) context of open
strings on D-branes in non-constant $B$-field backgrounds.

A related interpretation of the twisted Poisson structure
(\ref{luestalg}) on phase space $\Mcal=T^*M$, as discussed in~\cite[Section~2.5]{Mylonas2012},
is that of a higher Poisson structure on the 
multivector field algebra
$\vect^\bullet(M)= C^\infty(M,\bigwedge^\bullet
TM)$.

\subsection{Deformations of configuration space geometry}

In order to formulate a nonassociative theory of gravity on the
original configuration space $M$, we note that there is a natural
foliation of $\Mcal$ over $M$ defined by the (global)
sections $s_{\bar p}:M\to\Mcal$ with $s_{\bar p}(x)=(x,\bar
p)$ for $\bar p\in (\real^d)^*$; this simply means that we can realise
$M$ as any leaf of constant
momentum in phase space. We can therefore restrict functions on phase space $\underline{f}
\in
C^\infty(\Mcal)$ to $C^\infty(M)$ via the family of pullbacks
$f_{\bar p}:= s_{\bar p}^*\, \underline{f}$. Thus via pullback along
the bundle projection $\pi:\Mcal\to M$ we can transport functions on $M$ into the cotangent
bundle, perform the necessary nonassociative deformations on $\Mcal$,
and then pullback along $s_{\bar p}$ to obtain the desired nonassociative deformations of
the geometry of $M$; schematically, we can depict this foliation by quantizations of
the configuration manifold $M$ via the commutative diagrams
\beq
\xymatrix@=15mm{ C^\infty(\Mcal) \ \ar[r]^{\mathscr{Q}} & \ \widehat{C^\infty(\Mcal)} \ar[d]^{s_{\bar
      p}^*} \\
 C^\infty(M)\ar[u]^{\pi^*} \ \ar@{.>}[r]_{\mathscr{Q}_{\bar p}} & \ \widehat{ C^\infty(M)}
}
\eeq
with $s_{\bar p}^*\circ \pi^*=(\pi\circ s_{\bar p})^*=\Id$, where $\mathscr{Q}$ is the nonassociative quantization 
of $\Mcal$ whose foundations are
developed by~\cite{Mylonas2012,Mylonas2013}, while $\mathscr{Q}_{\bar
  p}$ is the desired nonassociative quantization of $M$ induced by the diagram. 
We can use this perspective to obtain
nonassociative field theories on $M$ via this systematic implementation
of the nonassociative geometry of $\Mcal$.

This perspective also suggests a way in which to obtain richer
deformations of the geometry of configuration space from the
nonassociative quantization of phase space. By pullbacks along $\pi^*$ we consider fields $\underline{f}
\in C^\infty(\Mcal)$ that satisfy the \emph{section constraints}
\beq
\tilde\partial^i\, \underline{f}=0 \qquad \mbox{for} \quad i=1,\dots,d \ ,
\label{eq:section}\eeq
where $\tilde\partial^i:= \frac\partial{\partial p_i}$. 
There is a natural invariant metric on phase space ${\cal M}=T^*
M$ which in a local frame is defined by
\beq
\gamma=\dd x^i\otimes \dd p_i +\dd p_i\otimes \dd x^i \ .
\label{eq:gamma}\eeq
It defines an $O(d,d)$-structure on $\Mcal$, i.e., a reduction of the structure group $GL(2d,\real)$ of the tangent bundle $T\Mcal$ to the subgroup $O(d,d)$; this is the symmetry group underlying quantum mechanical Born reciprocity and its relation to T-duality is explained by~\cite{Freidel2013}. 
Drawing now on the evident similarities
with double field theory (see~\cite{Aldazabal2013,Hohm2013} for
reviews and~\cite{Vaisman2012} for a rigorous mathematical treatment), a natural weakening of the section condition
(\ref{eq:section}) involves constraints on phase space functions
$\underline{f}$ which remove their momentum dependence 
locally only up to an $O(d,d)$ transformation.
For example, if
$\theta$ is any constant nondegenerate bivector on $M$, then this
constraint is solved by any phase space function $\underline{f}$ which
depends only on the non-local Bopp shifts $x+\theta\cdot p$; this
$d$-dimensional slice of the $2d$-dimensional phase space geometry can
be rotated to the constant momentum leaves above via an $O(d,d)$
transformation. 
In particular, the tangent bundle on phase space decomposes as $T\Mcal\cong L\oplus L^*$ where $L$ is the tangent bundle on the leaves of the foliation and $L^*$ is its dual bundle with respect to the orthogonal complement in the metric (\ref{eq:gamma}). The section constraints on functions $\underline{f}
\in
C^\infty(\Mcal)$ and more generally on covariant tensor fields $\underline{T}\in C^\infty(\Mcal,\bigotimes^\bullet T^*\Mcal)$ now read as
\beq
Z(\,\underline{f}\,)=0 \qquad \mbox{and} \qquad \imath_{Z}\,\underline{T}=0 = \Lcal_{Z}\,\underline{T}
\label{eq:gensection}\eeq
for all sections $Z\in C^\infty(\Mcal, L^*)$, where $\imath$ denotes contraction and $\Lcal$ is the Lie derivative. This means that the set of admissible fields is constrained to \emph{foliated} tensor fields with respect to the distribution $L^*$. For the foliation defined by $s_{\bar p}$, one has $L=TM$ and the section constraints can be imposed by taking $Z=\tilde\p^i$ for each $i=1,\dots,d$. A similar perspective on closed string target spaces is addressed in~\cite{Freidel2013} where connections with double field theory are also explored.

These more general foliations of phase space could lead to much richer classes of deformations of $M$, because the nonassociative
quantizations $\mathscr{Q}$ are \emph{not} $O(d,d)$-invariant, as is evident from the defining relations (\ref{luestalg}). Moreover, quantization of the section conditions (\ref{eq:gensection}) themselves could lead to interesting deformations of the foliated tensor fields on the $d$-dimensional slices. It would be interesting
to repeat the analysis of this paper for such more general section
slices. Among other things, this could help elucidate possible relationships between the
noncommutative gerbe structure~\cite{Aschieri:2002fq} on phase space underlying the
nonassociative deformations~\cite{Mylonas2012} and the abelian gerbe
structures underlying the generalised manifolds in double
geometry~\cite{Berman2014,Hull2014}. It should also help in
understanding how to lift geometric objects from $M$ to its cotangent
bundle $\Mcal=T^*M$ in a way suited to describe nonassociative
deformations of the geometry and of gravity directly on the
configuration manifold $M$.

\newsection{Triproducts from phase space star products\label{sec:star}}

\subsection{Families of $n$-triproducts}

We shall now derive triproduct formulas in the context of Section~\ref{sec:NAdef} which
include those of~\cite{Blumenhagen2011} as special cases. Our starting
point is the nonassociative star product $\star$ on phase space derived by~\cite{Mylonas2012}
(see also~\cite{Bakas2013}) which quantizes the
commutation relations (\ref{luestalg}). Here we shall utilize the
expression for this star product as a
twisted convolution product (see  Section \ref{sec:cst} for a different derivation). 
For phase space functions
$\underline{f} ,\underline{g} \in
C^\infty(\Mcal)$, the integral formula 
derived in ~\cite[eq.~(3.38)]{Mylonas2013} adapted to the normalizations in
(\ref{luestalg})  reads as
\bea
(\, \underline{f} \star \underline{g}\, )(x,p) &=& \frac1{(\pi\, \hbar)^{2d}}\, \int\!\!\!\int_M\, \dd^d z\
\dd^d z' \ \int\!\!\!\int_{M^*}\, \dd^d k\ \dd^d k' \
\underline{f}(x+z,p+k)\, \underline{g}(x+z', p+k'\,) \nonumber \\ &&
\hspace{5cm} \ \times \ \e^{-\frac{2\ii}\hbar\, (k\cdot
  z'-k'\cdot z)} \, \e^{-\frac{2\ii\ell_s^4}{3\hbar^3}\, R(k,k',p)} \ ,
\label{intformula}\eea
where $k'\cdot z:= k'_i\, z^i$ and $R(k,k',p):= R^{ijl}\,
k_i\, k'_j\, p_l$. Restricted to functions $f,g\in C^\infty(M)$ by the
pullback along $\pi:\Mcal\to M$, after rescaling, this product reads as
\bea
(\pi^*{f} \star \pi^*{g})(x,p) &=& \frac1{(2\pi)^{2d}}\, \int\!\!\!\int_M\, \dd^d z\
\dd^d z' \ {f}(x+z)\, {g}(x+z'\,) \nonumber \\ && \qquad \qquad
\times \ \int\!\!\!\int_{M^*}\, \dd^d k\ \dd^d k' \
\e^{-\ii(k\cdot
  z'-k'\cdot z)} \, \e^{-\frac{\ii\ell_s^4}{6\hbar}\, R(k,k',p)} \ .
\eea
In general, the star product of two fields on $M$ is a field on
$\Mcal$, i.e., a differential operator on $M$. Thus the star product
$\star$ does not close in the algebra $A=C^\infty(M)$. We can define
a family of products of fields on $M$ by pulling this product
back under the local sections $s_{\bar  p}:M\to \Mcal$ that foliate
${\cal M}$ with leaves isomorphic to $M$.  This defines a family of 2-products
$\mu^{(2)}_{\bar p}:A\otimes A\to A$ given
by
\bea
\mu^{(2)}_{\bar p} (f\otimes g)(x)&:=& s_{\bar p}^*\, (\pi^*f\star \pi^*g)(x,p) \nonumber \\[4pt] 
&=& \frac1{(2\pi)^{2d}}\, \int\!\!\!\int_M\, \dd^d z\
\dd^d z' \ {f}(x+z)\, {g}(x+z'\,) \nonumber \\ && \qquad \qquad \times
\ \int\!\!\!\int_{M^*}\, \dd^d k\ \dd^d k' \
\e^{-\ii(k\cdot
  z'-k'\cdot z)} \, \e^{-\frac{\ii\ell_s^4}{6\hbar}\, \theta_{\bar
    p}(k,k'\, )}
\nonumber\\[4pt]
&=& \frac1{(2\pi)^{2d}}\, \int\!\!\!\int_{M^*} \, \dd^dk\ \dd^dk' \
\e^{\ii (k+k'\,)\cdot x}\, \e^{\frac{\ii\ell_s^4}{6\hbar}\, \theta_{\bar
    p}(k,k'\,)} \nonumber \\
&& \qquad \qquad \times \ \int\!\!\!\int_M \, \dd^dz\ \dd^dz'\ f(z)\,
g(z'\,)\, \e^{-\ii(k\cdot z'+k'\cdot z)} 
\nonumber\\[4pt] 
&=& \int_{M^*}\, \dd^d k \ \int_{M^*}\, \dd^dk' \ \hat f(k'\,)\, \hat
g(k) \e^{\ii(k+k'\,)\cdot x}\, \e^{\frac{\ii\ell_s^4}{6\hbar}\, \theta_{\bar
    p}(k,k'\,)} \nonumber\\[4pt] 
&=& \int_{M^*}\, \dd^d k \ \int_{M^*}\, \dd^dk' \ \hat f(k)\, \hat
g(k'\,) \e^{\ii(k+k'\,)\cdot x}\, \e^{-\frac{\ii\ell_s^4}{6\hbar}\, \theta_{\bar
    p}(k,k'\,)} \ , 
\eea
where
\beq
\theta_{\bar p}=\mbox{$\frac12$}\, \theta^{ij}_{\bar
  p}\, \partial_i\wedge \partial_j :=
\mbox{$\frac12$} \, R^{ijk}\, \bar
p_k\, \partial_i \wedge \partial_j\label{defthetap}
\eeq
are constant bivectors on $M$
along the leaves of constant momentum in phase space, and
\beq
\hat f(k)=\frac1{(2\pi)^d}\, \int_M\, \dd^dx\ f(x)\, \e^{-\ii k\cdot x}
\eeq
is the Fourier transform of the function $f\in C^\infty(M)$. This last
expression is simply the momentum space representation of the associative
Moyal--Weyl star product $\As_{\bar p}$ determined by the bivector
$-\frac{\ell_s^4}{3\hbar}\, \theta_{\bar p}$, and we thus find an
expression for the 2-products in terms of a bidifferential operator
\beq\label{mu2bidif}
\mu^{(2)}_{\bar p}(f\otimes g) = f\, \As_{\bar p}\, g :=
\mu\Big(\exp\big(\mbox{$\frac{\ii\ell_s^4}{6\hbar}\, \theta_{\bar
    p}^{ij}\, \partial_i\otimes \partial_j$}\big)(f\otimes g)\Big) 
\eeq
where $\mu(f\otimes g)= f\, g$ is the pointwise multiplication of
functions in $A=C^\infty(M)$.

We shall now generalize this result to $n$-products $\mu_{\bar
  p}^{(n)}:A^{\otimes n}\to A$
of functions $f_1,\dots,f_n\in A$ for $n\geq3$, which we define in an analogous
way. However, we must keep in mind that the phase space star product
$\star$ is nonassociative, so we have to keep track of the order in
which we group binary products of functions on $\Mcal$; this is even
true after integration over $\Mcal$~\cite{Mylonas2013}. We choose the ordering
\beq
\mu_{\bar p}^{(n)}(f_1\otimes\cdots \otimes f_n):= s_{\bar
  p}^*\big[\big((\cdots(\pi^*f_1\star \pi^*f_2)\star \pi^*f_3)\star \cdots)\star
\pi^* f_n \big] \ ,
\label{eq:ntriproddef}\eeq
and discuss the significance of the other orderings below. We shall
compare this $n$-product to 
iterations of the Moyal--Weyl star products given by
\beq
 \As_{\bar p}(f_1\otimes \cdots\otimes f_n) :=
f_1\, \As_{\bar p}\, \cdots\, \As_{\bar p}\,  f_n = \mu\Big[\exp\Big(\,
\frac{\ii\ell_s^4}{6\hbar}\, \sum_{1\leq a<b\leq n}\, \theta_{\bar
    p}^{ij}\, \partial^a_i \;\partial^b_j\,
  \Big)(f_1\otimes\cdots\otimes f_n)\Big] \ ,
\label{eq:nmoyaldef}
\eeq
where $\partial_i^a=(\Id\otimes \cdots \otimes \partial_i
\otimes\cdots\otimes \Id)$, $a=1,\dots,n$  denotes the derivative
$\partial_i$ acting on the $a$-th factor of the tensor product
$f_1\otimes\cdots\otimes f_n$; note that no bracketings need be
specified here due to associativity of the
products $\As_{\bar p}$ for all $\bar p$. 
\begin{proposition} 
$$
\mu_{\bar p}^{(n)}(f_1\otimes\cdots\otimes f_n) = \As_{\bar
  p}\Big[\exp\Big(\, \frac{\ell_s^4}{12}\ \sum_{1\leq a<b<c\leq n}\,
R^{ijk}\, \partial_i^a\, \partial_j^b\, \partial_k^c\, \Big)
(f_1\otimes\cdots\otimes f_n)\Big] \ .
$$
\label{prop:ntriproduct}\end{proposition}
\begin{proof}
By iterating the integral formula (\ref{intformula}), it is
straightforward to derive the integral representation
\bea
&& s^*_{\bar p}\big[\big((\cdots(\pi^*f_1\star \pi^*f_2)\star \pi^*f_3)\star \cdots)\star
\pi^* f_n \big](x) \nonumber \\ && \qquad \qquad \ = \ \frac1{(\pi\,
  \hbar)^{2(n-1)d}}\, \prod_{a=1}^{n-1}\ \int\!\!\!\int_{M}\, \dd^d z_a\
  \dd^d z_a'\ \int\!\!\!\int_{M^*}\, \dd^d k_a\ \dd^d k_a'\
  \e^{-\frac{2\ii}\hbar\, (k_a\cdot z_a'-k_a'\cdot z_a)} \nonumber\\
&&\qquad \qquad \qquad \qquad \times \ \prod_{b=1}^n\,
f_b(x+z_1+\cdots+z_{n-b}+ z'_{n-b+1}) \\ && \qquad
\qquad \qquad \qquad \qquad \qquad \times \ 
\exp\Big(\,-\frac{2\ii\ell_s^4}{3\hbar^3}\ \sum_{c=1}^{n-1}\,
R(k_c,k_c',p+k_{1}+\cdots + k_{c-1})\, \Big) \nonumber
\eea
with the conventions $z_0=0, z_n'=0, k_0=0$. After rescaling we thus get
\bea
\mu_{\bar p}^{(n)}(f_1\otimes\cdots\otimes f_n) (x) &=&
\frac1{(2\pi)^{2(n-1)d}}\ \prod_{a=1}^{n-1}\ \int\!\!\!\int_{M}\, \dd^d
  z_a\ 
  \dd^d z_a'\ \int\!\!\!\int_{M^*}\, \dd^d k_a\ \dd^d k_a'\
  \e^{-\ii (k_a\cdot z_a'-k_a'\cdot z_a)} \nonumber\\
&&\qquad \qquad \times \ \prod_{b=1}^n\,
f_b(x+z_1+\cdots z_{n-b}+ z'_{n-b+1}) \nonumber\\
&& \qquad \qquad \qquad \qquad \times \
\exp\Big(-\frac{\ii\ell_s^4}{12}\ \sum_{c=1}^{n-1}\,
R(k_c,k_c',k_{1}+\cdots + k_{c-1})\Big) \nonumber\\
&& \qquad \qquad \qquad \qquad \qquad \qquad \times \
\exp\Big(-\frac{\ii\ell_s^4}{6\hbar}\ \sum_{c=1}^{n-1}\, \theta_{\bar
  p}(k_c,k_c'\, ) \Big) \nonumber\\[4pt]
&=& \frac1{(2\pi)^{2(n-1)d}}\ \prod_{a=1}^{n-1}\ \int\!\!\!\int_{M}\, \dd^d z_a\
  \dd^d z_a'\ \int\!\!\!\int_{M^*}\, \dd^d k_a\ \dd^d k_a'\
  \e^{-\ii (k_a\cdot z_a'-k_a'\cdot z_a)} \nonumber\\
&& \qquad \qquad \times \ \e^{\ii(k_1+\cdots+k_{n-1}-k'_{n-1})\cdot x}\,
\e^{-\ii k'_{n-1}\cdot(z_1+\cdots + z_{n-2})}\nonumber \\ 
&& \qquad \qquad
\qquad \times \e^{\ii
  k_a\cdot (z_{1}+\cdots+z_{a-1})} \ \prod_{b=2}^n\,
f_b(z'_{n-b+1}) \ f_1(z_{n-1}\,) \nonumber \\
 && \qquad \qquad \qquad \qquad \times \ 
\exp\Big(-\frac{\ii\ell_s^4}{12}\ \sum_{c=1}^{n-1}\,
R(k_c,k_c',k_{1}+\cdots + k_{c-1})\Big) \nonumber\\
&& \qquad \qquad \qquad \qquad \qquad \qquad \times \
\exp\Big(-\frac{\ii\ell_s^4}{6\hbar}\ \sum_{c=1}^{n-1}\, \theta_{\bar
  p}(k_c,k_c'\, ) \Big) \ ,
\eea
where we relabelled 
$z'_{n-b+1}\to z'_{n-b+1}-(x+z_1+\cdots + z_{n-b})$ for $2\leq b\leq
n$, and $z_{n-1}\to z_{n-1}-(x+z_1+\cdots+z_{n-2})$. Integrating over
$z_a$ for $a=1,\dots,n-2$ gives delta-function constraints
\beq
k'_a=k'_{n-1}-(k_{a+1}+\cdots +k_{n-1}) \qquad \mbox{for} \quad a=1,\dots,n-2
\ .
\eeq
Integrating over $z'_a$ for $ a=1,\dots,n-1$ and over $z_{n-1}$, and
after relabelling $k_1\to k_n$, $k'_{n-1}\to -k_1$ and $k_{n-b+1}'\to k_{b}$ for $b=2,\dots,n$, we then get
\bea
\mu_{\bar p}^{(n)}(f_1\otimes\cdots\otimes f_n) (x) 
&& \!\!\!\!\!\!\!= \ \prod_{a=1}^{n}\
\int_{M^*}\, \dd^d k_a\ \hat f_a(k_a) \ \e^{\ii (k_1+\cdots+k_n)\cdot x} \nonumber \\
&& \qquad \times \ 
\exp\Big(\frac{\ii\ell_s^4}{12}\, 
\sum_{c=1}^{n-1}\, R(k_{n-c+1},k_1+\cdots 
+k_{n-c},k_{n-c+2}+\cdots +k_n)
\Big) \nonumber\\
&& \qquad \qquad \times \
\exp\Big(\frac{\ii\ell_s^4}{6\hbar}\ \sum_{c=1}^{n-1}\, \theta_{\bar
  p}(k_{n-c+1},k_1+\cdots 
+k_{n-c} \, ) \Big) \nonumber\\
&& \!\!\!\!\!\!\!=  \ \prod_{a=1}^{n}\
\int_{M^*}\, \dd^d k_a\ \hat f_a(k_a) \ \e^{\ii (k_1+\cdots+k_n)\cdot x} \nonumber \\
&& \qquad \times \ 
\exp\Big(-\frac{\ii\ell_s^4}{12}\, 
\sum_{c=1}^{n-1} \ \sum_{a=1}^{n-c} \ \sum_{~b=n-c+2}^n\, R(k_a,k_{n-c+1},k_b)
\Big)
\nonumber\\
&& \qquad \qquad \times \
\exp\Big(-\frac{\ii\ell_s^4}{6\hbar}\
\sum_{c=1}^{n-1} \ \sum_{a=1}^{n-c}\, 
\theta_{\bar p}(k_a,k_{n-c+1} \, ) \Big) 
\eea
where we used  multilinearity and antisymmetry of the trivector and
bivector terms.
Rewriting the summations we finally
arrive at
\bea
\mu_{\bar p}^{(n)}(f_1\otimes\cdots\otimes f_n)(x) &=& \prod_{a=1}^{n}\
\int_{M^*}\, \dd^d k_a\ \hat f_a(k_a) \ \e^{\ii (k_1+\cdots+k_n)\cdot x} \nonumber \\
&& \qquad \qquad \times \ 
\exp\Big(-\frac{\ii\ell_s^4}{12}\, \sum_{1\leq a<e<b\leq n}\,
R(k_a,k_e,k_b) 
\Big) \nonumber \\
&& \qquad \qquad \qquad \qquad \times \ \exp\Big(-\frac{\ii\ell_s^4}{6\hbar}\, \sum_{1\leq a<b\leq
  n}\, \theta_{\bar p}(k_a,k_b) \Big) \ ,
\label{eq:ntriprodmom}\eea
which is the asserted result written using the momentum space
representation for the $\As_{\bar p}$ product and the $R$-flux
differential operators.
\end{proof}

Besides the fundamental 2-products $\mu_{\bar p}^{(2)}$, which we
have seen coincide with the Moyal--Weyl star products $\As_{\bar p}$,
the basic triproducts are given by
\beq\label{basic3trip}
\mu_{\bar p}^{(3)}(f\otimes g\otimes h) = \As_{\bar
  p}\Big(\exp\big(\, \mbox{$\frac{\ell_s^4}{12}$} \, 
R^{ijk}\, \partial_i\otimes \partial_j\otimes \partial_k\, \big)
(f\otimes g\otimes h)\Big) \ .
\eeq
The coordinate space commutator in (\ref{luestalg}) is then reproduced
by the commutator bracket for $\mu_{\bar p}^{(2)}$, while the 3-bracket (\ref{eq:3bracket}) is reproduced by
$[x^i,x^j,x^k]_{\bar p}$, where
\beq
[f_1,f_2,f_3]_{\bar p}:= \sum_{\sigma\in S_3}\, (-1)^{|\sigma|}\,
\mu_{\bar p}^{(3)}(f_{\sigma(1)}\otimes f_{\sigma(2)}\otimes f_{\sigma(3)})
\eeq
with $S_3$ the group of permutations of the set $\{1,2,3\}$.
The general $n$-triproducts satisfy the reduction properties
\beq
\mu_{\bar
  p}^{(n)}\big(f_1\otimes\cdots\otimes(f_i\!=\!1)\otimes\cdots\otimes
f_{n} \big)\,=\,\mu_{\bar 
  p}^{(n-1)}\big(f_1\otimes\cdots\otimes \widehat{f_i}\otimes\cdots
\otimes f_{n} \big) 
\eeq
where $\widehat{f_i}$ denotes omission of $f_i$, $i=1,\ldots, n$; these reductions consistently yield
\beq
\mu_{\bar p}^{(3)}(f\otimes g\otimes 1) = 
\mu_{\bar p}^{(3)}(f\otimes 1\otimes g)=\mu_{\bar p}^{(3)}(1\otimes
f\otimes g)=
\mu_{\bar p}^{(2)}(f\otimes g)=f\, \As_{\bar p}\,  g \ . 
\eeq
However, for $n>2$ the $n$-triproducts $\mu_{\bar p}^{(n)}$
cannot be defined by iteration from $m$-triproducts with
$m<n$. This is in contrast to the precursor definition
(\ref{eq:ntriproddef}) in terms of phase space star products, and also
to the $n$-star products which are defined in (\ref{eq:nmoyaldef})
by iteration of the Moyal--Weyl star products. 

One of the most distinctive features of the $n$-triproducts is 
their trivialization on-shell, i.e., after integration over $M$.
\begin{proposition}
 \ $\displaystyle{
\int_M\, \dd^d x\ \mu_{\bar p}^{(n)}(f_1\otimes \cdots\otimes f_n) = \int_M\, \dd^d x\ f_1 \,\As_{\bar
  p}\, \cdots \, \As_{\bar p}\,  f_n \ .
}$
\label{prop:onshell}\end{proposition}
\begin{proof}
Using the momentum space integral formula (\ref{eq:ntriprodmom}) for
the $n$-triproduct, we have
\beq
\int_M\, \dd^d x\ \mu_{\bar p}^{(n)}(f_1\otimes\dots\otimes f_n) &=& (2\pi)^d\
\prod_{a=1}^{n}\
\int_{M^*}\, \dd^d k_a\ \hat f_a(k_a) \ \delta(k_1+\cdots+k_n) \nonumber \\
&& \qquad \qquad \qquad \times \ 
\exp\Big(- \frac{\ii\ell_s^4}{12}\, \sum_{1\leq a<b<c\leq n}\,
R(k_a,k_b,k_c) 
\Big) \\
&& \qquad \qquad \qquad \qquad \qquad \times \ \exp\Big(-\frac{\ii\ell_s^4}{6\hbar}\, \sum_{1\leq a<b\leq
  n}\, \theta_{\bar p}(k_a,k_b) \Big) \ , \nonumber
\eeq
where the delta-function enforcing momentum conservation arises by
translation-invariance on $M$ of the definition
(\ref{eq:ntriproddef}). Using multilinearity
of the trivector $R$ we have
\bea
\sum_{1\leq a<b<c\leq n}\,
R(k_a,k_b,k_c) = \sum_{b=2}^{n-1} \ \sum_{a=1}^{b-1} \
\sum_{c=b+1}^{n}\, R(k_a,k_b ,k_c) = \sum_{b=2}^{n-1}\, R\Big(\, \mbox{$\sum\limits_{a=1}^{b-1}\, 
k_a\,,\,k_b\,,\,\sum\limits_{c=b+1}^{n}\, k_c$} \, \Big) \ ,
\eea
and each term in the sum over $b$ vanishes by antisymmetry of $R$
after imposing momentum conservation $\sum_{a=1}^n \, k_a=0$. 
\end{proof}

From Proposition~\ref{prop:onshell} it follows that the leaf of zero momentum
is singled out by the feature that on-shell there
are no signs of noncommutativity or nonassociativity. On the special
slice $\bar p=0$ the
2-product
\beq
\mu_0^{(2)}(f\otimes g) = f\, g
\eeq
is the ordinary multiplication of functions, while for $n\geq3$ the
expression for the $n$-triproduct from
Proposition~\ref{prop:ntriproduct} becomes
\beq
\mu_{0}^{(n)}(f_1\otimes\cdots\otimes f_n) = \mu\Big[\exp\Big(\, \frac{\ell_s^4}{12}\ \sum_{1\leq a<b<c\leq n}\,
R^{ijk}\, \partial_i^a\, \partial_j^b\, \partial_k^c\, \Big)
(f_1\otimes\cdots\otimes f_n)\Big] \ .
\eeq
These products coincide with the sequence of $n$-triproducts for $n\geq2$ that
was proposed by~\cite{Blumenhagen2011} from an analysis of off-shell
closed
string tachyon amplitudes in the toroidal flux model to linear order
in the flux components $R^{ijk}$. By Proposition~\ref{prop:onshell} they obey
the on-shell condition
\beq
\int_M\, \dd^d x\ \mu_{0}^{(n)}(f_1\otimes\cdots\otimes f_n) = \int_M\, \dd^d
x\ f_1 \cdots  f_n \ ,
\eeq
and hence all traces of nonassociativity disappear in closed string
scattering amplitudes. Here we have reproduced this conjectural
triproduct to all orders in $R^{ijk}$ from the all orders phase space
star product derived by~\cite{Mylonas2012}. The correspondence between
the phase space star product restricted to functions on $M$ and the
triproducts of~\cite{Blumenhagen2011} was already noted
by~\cite{Mylonas2012} for $n=2,3$ (see also~\cite{Bakas2013}); here we
have extended and generalised the correspondence to all $n>3$.

The triproducts for $\bar p\neq0$ lead to on-shell correlation
functions which are cyclically invariant, by cyclicity of the
Moyal--Weyl products $\As_{\bar p}$ for all $\bar p$. In this case the
scattering amplitudes resemble those of \emph{open} strings on a
D-brane with two-form $B$-field inverse to the bivector $\theta_{\bar
  p}$. The appearence of the more general triproducts $\mu_{\bar
  p}^{(n)}$ is natural from the perspective of the open/closed string
duality described in~\cite[Section~2.4]{Mylonas2012}, which is a
crucial ingredient in the derivation of the nonassociative phase space
star product from closed string correlation functions; in fact,
in~\cite[Appendix]{Lust2010} it is argued that closed string momentum
and winding modes define a certain notion of D-brane in closed string theory. Moreover, they
are natural in light of the momentum space noncommutative gerbe
structure of the nonassociative phase space description of the
parabolic $R$-flux background and the closed string Seiberg--Witten maps relating
associative and nonassociative theories, as described
in~\cite[Section~3.4]{Mylonas2012}.  Hence in the following we work
with the general leaves of constant momentum in order to retain as much
of the (precursor) phase space nonassociative geometry on $M$ as
possible in our analysis, with the understanding that bonafide closed string
scattering amplitudes require setting $\bar p=0$ at the end of the
day. We can interpret this distinction in the following way: If we
regard the nonassociative field theory constructed on phase space as
the fundamental field theory of closed strings in the non-geometric
$R$-flux frame (as was done in~\cite{Mylonas2012} and as we do
throughout in the present paper), with the leaves of constant momentum
being the on-shell physical sectors, then integrating over momenta
traces out extra degrees of freedom leading to violation of
associativity. From this perspective localizing on the $\bar p=0$ leaf
corresponds to working on a closed string vacuum as a lowest
order approximation to the theory, while higher (M2-brane) excitations involve fluctuations out of the $\bar p=0$ leaf and probe the whole structure of the full nonassociative field theory. 

\subsection{Association relations}

In our definition (\ref{eq:ntriproddef}) we used a particular
bracketing for the star product of $n$ functions on phase space. As
this star product is nonassociative, it is natural to ask what happens
when one chooses different orderings. It was shown by~\cite{Mylonas2013}
that the different choices of associating
the functions is controlled by an associator, which can
be described by a tridifferential operator
\beq\label{phifgh}
\Phi=\exp\big(\mbox{$\frac{\ell_s^4}6$} \,
R^{ijk}\, \partial_i\otimes \partial_j\otimes\partial_k\big) 
\label{assoctridiff}~.\eeq
The associator relates the two different products of three functions on phase space
$(\,\underline{f}\star \underline{g}\,)\star \underline{h}$ and
$\underline{f}\star (\,\underline{g}\star \underline{h}\,)$.
If we introduce the notation
\eq
\Phi(\,\uf\otimes\ug\otimes\uh\,)=:\uf^\phi\otimes\ug^\phi\otimes\uh^\phi
\ee
where summation over the index $\phi$ is understood, then
the relation is
\eq
(\,\uf\star\ug\,)\star \uh=\uf^\phi\star(\,\ug^\phi\star\uh^\phi\,)~. \label{id329}
\ee
We can write this relation more algebraically as
\beq
\mu_\star\circ(\mu_\star\otimes\Id)=\mu_\star\circ(\Id\otimes\mu_\star)
\circ\Phi
\ ,
\label{algid329}\eeq
where $\mu_\star(\,\uf\otimes \ug\,):=\uf\star\ug$ with the linear maps
\beq
\mu_\star\circ(\mu_\star\otimes\Id):(\alg\otimes 
\alg)\otimes  \alg\longrightarrow 
\alg \qquad \mbox{and} \qquad
\mu_\star\circ(\Id\otimes\mu_\star): \alg\otimes (
\alg\otimes \alg)\longrightarrow 
\alg \ .
\eeq
Here we regard the associator as a linear map
\eq\label{asscat}
\Phi:(\alg\otimes \alg)\otimes \alg\longrightarrow \alg\otimes
(\alg\otimes \alg) \qquad \mbox{with} \quad (\,\uf\otimes \ug\,)\otimes
\uh\stackrel{\Phi} \longmapsto \uf^\phi\ot(\,\ug^\phi\ot\uh^\phi\,) \ ,
\ee
where the parentheses emphasize that on the left-hand side the product $\mu_\star\circ(\mu_\star\otimes\Id)$ naturally
acts while on the right-hand side the
product $\mu_\star\circ(\Id\otimes \mu_\star)$ acts.  Applying these
natural products to (\ref{asscat}) yields the identity (\ref{id329}).

The operator $\Phi$ is a 3-cocycle in the Hopf algebra of translations
and Bopp shifts 
in phase space $\Mcal$ that we will discuss in
Section~\ref{sec:cochain}; this means that it obeys the pentagon
relations which states that the two possible ways of reordering the
brackets from left to right 
$\big((\, \alg\, \otimes\, \alg\, )\otimes\,
\alg\, \big)\otimes\, \alg\rightarrow  \alg\, \otimes \big(\, \alg\, \otimes
(\, \alg\, \otimes\, \alg\, ) \big) $ are equivalent, i.e.,
 the diagram
\beq
\xymatrix{
 & (\, \alg\, \otimes\, \alg\, )\otimes(\,
 \alg\, \otimes\, \alg\, ) \ar[ddr]^{\,\,\,\,\Phi_{\,\alg\,,\,\alg\,,\,
     \alg\,\otimes\, \alg}}
 & \\
 & & \\\ar[uur]^{\Phi_{\alg\, \otimes \,\alg\,,\,
     \alg\,,\, \alg\,\,\, }\,\,\,}
\big((\, \alg\, \otimes\, \alg\, )\otimes\,
\alg\, \big)\otimes\, \alg\ar[dd]_{\Phi_{\,\alg\,,\,\alg\,,\, 
     \alg}\, \otimes\, \Id} 
   & & \alg\, \otimes \big(\, \alg\, \otimes
(\, \alg\, \otimes\, \alg\, ) \big) \\
 & & \\
\big(\, \alg\, \otimes (\, \alg\, \otimes\,
\alg\,) \big)\otimes\, \alg\ 
\ar[rr]^{\Phi_{\,\alg\,,\,\alg\, \otimes \,
     \alg\,,\, \alg}}& & \ \alg\, \otimes
   \big(( \, \alg\, \otimes\,
\alg\,) \otimes\, \alg\,  \big) \ar[uu]_{\Id\, \otimes
  \, \Phi_{\,\alg\,,\,
     \alg\,,\, \alg}} \!\!\!\!\!\!\!\!
}
\label{eq:pentagon}
\eeq
commutes.
This means that we can rewrite the products $((\,\uf\star\ug\,
)\star\uh\, )\star\uk$ as linear
combinations of products of the kind
$\uf'\star(\, \ug'\star(\, \uh'\star\uk'\, ))$. There are \emph{a priori} two different
linear combinations, respectively obtained by following the upper and
lower paths in the
diagram (\ref{eq:pentagon}). Commutativity of the diagram asserts
that these two paths are equivalent.
The action of the differential operator $\Phi_{\alg\otimes \alg\,,\,  \alg\,,\,
  \alg}$ on $(\, \uf\ot \ug\, )\ot \uh\ot \uk$ is inherited from the Leibniz rule
$\partial_i(\,\uf\ot \ug\,)=\partial_i\,\uf\ot \ug+\uf\ot \partial_i\,
\ug$ and it reads as
\eq
\Phi_{\alg\otimes \alg\,,\,  \alg\,,\, \alg}\big((\, \uf\ot\ug\, )\ot\uh\ot \uk\, \big)= \exp\big(\mbox{$\frac{\ell_s^4}{6}$}\, R^{ijk}\,
(\partial_i^1\, \partial_j^3\, \partial_k^4+
\partial_i^2\,\partial_j^3\,\partial_k^4) \big)(\, \uf\otimes
\ug\otimes \uh\otimes \uk\, ) \ .
\ee
If  we apply the product
$\mu_\star\circ(\mu_\star\otimes\mu_\star)$ to this expression we
obtain
\beq
\mu_\star\circ(\mu_\star\otimes\mu_\star)\circ \Phi_{\alg\ot
  \alg\,,\,\alg\,,\,\alg}\big((\,\uf\ot \ug\, )\ot \uh\ot \uk\, \big)
=\big((\,\uf\star\ug\,)\star\uh\, \big)\star\uk \ ,
\eeq
as is read off from the upper arrow in (\ref{eq:pentagon}). The differential operators $\Phi_{\alg\,,\,  \alg\,,\,
  \alg\ot \alg}$ and  $\Phi_{\alg\,,\,  \alg\ot \alg\,,\alg}$ are similarly defined.

The associator $\Phi$ also relates the two inequivalent triproducts of
functions on  $M$ given by
$s_{\bar p}^*\big[(\pi^*f\star\pi^*g) \star\pi^*h\big]$ and 
$s_{\bar  p}\big[\pi^*f\star(\pi^*g\star \pi^*h)\big] $.
In order to present an explicit expression for the product $s_{\bar
  p}\big[\pi^*f\star(\pi^*g\star \pi^*h)\big] $ we observe that there
is an obvious one-to-one correspondence between differential operators
$\partial_i$ on $\alg= C^\infty({\cal M})$ and on $A=C^\infty({M})$, and that for
any function $f$ on $M$ one has
$s_{\bar p}(\partial_i(\pi^\ast f))=\partial_i f$. This implies that the
associator $\Phi$ naturally acts as a differential operator on
$A\ot A\ot A$. 
We define
\eq
\Phi^{-1}(\, \uf\otimes\ug\otimes\uh\,
)=\uf^{\bar{\phi}}\otimes\ug^{\bar{\phi}}\otimes\uh^{\bar{\phi}}
\qquad \mbox{
  and } \qquad \Phi^{-1}(f\otimes g\otimes h)= f^{\bar{\phi}}\otimes
g^{\bar{\phi}}\otimes h^{\bar{\phi}}
\ee
with implicit summation as before, and then we find
\bea
s^\ast_{\bar  p}\big[\big(\pi^*f\star(\pi^*g\star
\pi^*h)\big)\big]&=& 
s_{\bar p}^*\big[\big((\pi^*f)^{\bar\phi}\star (\pi^*g)^{\bar\phi}
\big)\star (\pi^*h)^{\bar\phi}\,\big]\nonumber \\[4pt]
&=& 
s_{\bar p}^*\big[\big((\pi^*f^{\bar\phi})\star (\pi^*g^{\bar\phi})
\big)\star (\pi^*h^{\bar\phi}\, )\big]\nonumber \\[4pt]
&=&
\mu_{\bar p}^{(3)}\big(f^{\bar\phi}\otimes g^{\bar\phi}\otimes
h^{\bar\phi}\, \big) \label{muphi-1}\\[4pt]
&=& \As_{\bar
  p}\Big(\exp\big(\, \mbox{$\frac{\ell_s^4}{12}$} \, 
R^{ijk}\, \partial_i\otimes \partial_j\otimes \partial_k\, \big) \,
\Phi^{-1}(f\otimes g\otimes h)\Big) \nonumber \\[4pt]
&=& \As_{\bar
  p}\Big(\exp\big(-\mbox{$\frac{\ell_s^4}{12}$} \, 
R^{ijk}\, \partial_i\otimes \partial_j\otimes \partial_k \big)
(f\otimes g\otimes h)\Big) \ .\nn
\eea
Thus we can generate the new triproduct $s^*_{\bar  p}\big[\big(\pi^*f\star(\pi^*g\star
\pi^*h)\big)\big] $ of
fields by applying the inverse of the associator $\Phi$ directly to functions
on $M$ as in (\ref{muphi-1}). This new triproduct also obeys
Proposition~\ref{prop:onshell}, since from  (\ref{muphi-1}) we have
\beq\label{inttriv}
\int_M\, \dd^dx ~s^\ast_{\bar  p}\big[\big(\pi^*f\star(\pi^*g\star
\pi^*h)\big)\big]
=
\int_M\,\dd^dx~\mu_{\bar p}^{(3)} \big[\Phi^{-1} (f\ot g\ot h) \big] =
\int_M\, \dd^dx\ f\, \As_{\bar p}\, g\, \As_{\bar p}\, h \ .
\eeq
A quick way to prove this is to note that
\beq
\Phi^{-1}\big(\e^{\ii k_1\cdot x}\otimes \e^{\ii k_2\cdot x}\otimes
\e^{\ii k_3\cdot x}\big) = \e^{\frac{\ii\ell_s^4}{6}\,
  R(k_1,k_2,k_3)} \ \big( \e^{\ii k_1\cdot x}\otimes \e^{\ii k_2\cdot x}\otimes
\e^{\ii k_3\cdot x} \big) \ ,
\eeq
and the extra phase factor is unity after imposing momentum
conservation $k_1+k_2+k_3=0$.

We can now generate all induced 4-triproducts from the
diagram (\ref{eq:pentagon}). For example one has
\bea
s^*_{\bar p}\big[(\pi^* f\star\pi^*
g)\star(\pi^*h\star\pi^*k)] &=& \mu^{(4)}_{\bar p}\big[\Phi^{-1}_{A\ot
A\,, \,A\,,\,A}\big( (f\ot g)\ot h\ot k\big)\big] \\[4pt]
&=&\mu^{(4)}_{\bar p}
\big[\exp\big(-\mbox{$\frac{\ell_s^4}{6}$}\, R^{ijk}\,
(\partial_i^1\, \partial_j^3\, \partial_k^4+
\partial_i^2\,\partial_j^3\,\partial_k^4)\big)(f\otimes
g\otimes h\otimes k)\big)\big] \nn\\[4pt]
&=&\As_{\bar p}
\big[\exp\big(\mbox{$\frac{\ell_s^4}{12}$}\, R^{ijk}\,
(\partial_i^1\, \partial_j^2\, \partial_k^3+
\partial_i^1\,\partial_j^2\,\partial_k^4
\nn\\ && \qquad \qquad \qquad \qquad -\, \partial_i^1\, \partial_j^3\, \partial_k^4-
\partial_i^2\,\partial_j^3\,\partial_k^4)\big) (f\otimes
g\otimes h\otimes k)\big]  \ . \nn
\eea
A completely analogous calculation following both upper arrows in the
diagram (\ref{eq:pentagon}) gives
\bea\label{eq:4triprodop}
s^*_{\bar p}\big[\pi^* f\star \big(\pi^*
g\star(\pi^*h\star\pi^*k)\big) \big] &=& \mu^{(4)}_{\bar p}\big[\Phi^{-1}_{A\ot
A\,, \,A\,,\,A}\,\Phi^{-1}_{A\,,\,A\,, \,A\ot A}\big( f\ot g\ot (h\ot
k)\big)\big] \\[4pt]
&=&\As_{\bar p}
\Big[\exp\Big(-\frac{\ell_s^4}{12} \, R^{ijk}\, \sum_{1\leq
a<b<c\leq 4}
\partial_i^a\, \partial_j^b\, \partial_k^c \,\Big)(f\otimes
g\otimes h\otimes k)\Big] \ . \nn
\eea
These new 4-triproducts serve just as well for describing the products
among off-shell
closed string tachyon vertex operators, and indeed we find as in
Proposition~\ref{prop:onshell} the on-shell result
\bea
\int_M\, \dd^d x \ s^*_{\bar p}\big[(\pi^* f\star\pi^*
g)\star(\pi^*h\star\pi^*k) \big] &=&
\int_M\, \dd^d x \ \mu^{(4)}_{\bar p}\big[\Phi^{-1}_{A\otimes A\,,\, A\,,\, A}
\big((f\otimes g)\ot h\ot k\big)\big]\nn\\[4pt] &=& 
\int_M\, \dd^d x  \ \As_{\bar p}\big[\Phi^{-1}_{A\otimes A\,,\, A\,,\, A}
\big((f\otimes g)\ot h\ot k\big)\big]\nn\\[4pt]  &=& 
\int_M \, \dd^d x  \;\As_{\bar p}\big[\Phi^{-1}
\big((f\, \As_{\bar p}\, g)\ot h\ot k \big)\big] \nn\\[4pt]   &=&
\int_M\, \dd^d x \ f \, \As_{\bar p}\, g\,
\As_{\bar p}\, h\, \As_{\bar p}\, k
\label{eq:on-shell-tachyon}\eea
where in the last equality we used (\ref{inttriv}). Letting
$R^{ijk}\to -R^{ijk}$ in Proposition~\ref{prop:onshell} we also
immediately see that the integral of (\ref{eq:4triprodop}) coincides
with (\ref{eq:on-shell-tachyon}).

The situation is notably different if one follows the left vertical
arrow in the diagram (\ref{eq:pentagon}). We first compute the product
\bea
s^*_{\bar p}\big[\big(\pi^* f\star (\pi^*
g\star\pi^*h) \big)\star\pi^*k\big] &=&
\mu_{\bar p}^{(4)}\big[\Phi^{-1}_{A\,,\,A\ot A\,,\, A} \big( f\ot (g\ot
h)\ot k\big) \big]\\[4pt]
&=& \As_{\bar p}\Big[\exp\big(\, \mbox{$\frac{\ell_s^4}{12}$}\, R^{ijk}\,
(-\partial_i^1\, \partial_j^2\, \partial_k^3+
\partial_i^1\,\partial_j^2\, \partial_k^4 \nn \\ 
&& \qquad \qquad +\, \partial_i^1\, \partial_j^3\, \partial_k^4
+ \partial_i^2\, \partial_j^3\, \partial_k^4) \big)(f\otimes
g\otimes h\otimes k)\Big] \ . \nonumber
\label{eq:1Phighk}\eea
Analogously we find for the product in the bottom right corner of
(\ref{eq:pentagon}) as
\bea
s^*_{\bar p}\big[\pi^* f\star \big((\pi^*
g\star\pi^*h)\star\pi^*k \big)\big] &=&
\mu_{\bar p}^{(4)}\big[(\Phi^{-1}\ot \Id)\,\Phi^{-1}_{A\,,\,A\ot A\,,\, A}\big( f\ot (g\ot h)\ot k\big)\big]\\[4pt]
&=&
\As_{\bar p}\Big[\exp\big(\mbox{$\frac{\ell_s^4}{12}$}\, R^{ijk}\,
(-\partial_i^1\, \partial_j^2\, \partial_k^3-
\partial_i^1\, \partial_j^2\, \partial_k^4 \nn\\ 
&&\qquad \qquad ~-\, \partial_i^1\, \partial_j^3\, \partial_k^4
+\partial_i^2\, \partial_j^3\, \partial_k^4) \big)(f\otimes
g\otimes h\otimes k)\Big] \ ,  \nonumber
\eea
and it is easy to see that a further application of
$\Id\ot \Phi^{-1}$ leads exactly to  the 4-triproduct
(\ref{eq:4triprodop}). As previously the application of
$\Phi^{-1}_{A\,,\,A\ot A\,,\, A}$ does not alter closed string amplitudes,  but
the application of $\Phi^{-1}\ot \Id$ does, and
we find
\bea
\int_M\, \dd^d x~ s^*_{\bar p}\big[\big(\pi^* f\star (\pi^*
g\star\pi^*h) \big)\star\pi^*k\big]&=& 
\int_M\, \dd^d x  ~\mu_{\bar 
  p}^{(4)}\big[ \Phi^{-1} (f\ot g\ot h)\ot
k \big]\\[4pt]
 &=& 
\int_M\, \dd^d x  ~\mu_{\bar 
  p}^{(4)}\big( f^{\bar\phi}\ot g^{\bar\phi}\ot h^{\bar\phi}\ot
k\big)\nn\\[4pt]
&=& \int_M\, \dd^d x~s^*_{\bar p}\big[\pi^* f\star \big((\pi^*
g\star\pi^*h)\star\pi^*k \big)\big] ~.\nn
\eea
In general one has
\bea
\int_M\, \dd^d x \ \mu_{\bar 
  p}^{(4)} \big( f^{\bar\phi}\ot g^{\bar\phi}\ot h^{\bar\phi}\ot
k\big)&=& \int_M \, \dd^d x \ f^{\bar\phi}\, \As_{\bar p} \,
g^{\bar\phi}\, \As_{\bar
  p}\, h^{\bar\phi}\, \As_{\bar p }\, k \nn \\[4pt] &\not=& \int_M\, \dd^d x \ f\,
\As_{\bar p} \, g\, \As_{\bar
  p} \, h\, \As_{\bar p }\, k \ .
\eeq
As a consequence these two 4-triproducts leave traces of
nonassociativity through additional interaction terms, and should hence
not be considered as physically viable products
among off-shell
closed string tachyon vertex operators.

This line of reasoning can be extended to all $n$-triproducts with
$n>4$. By MacLane's coherence theorem, all possible bracketings of
star products of $n$ functions on $\Mcal$ are related by
successive applications of the inverse associator $\Phi^{-1}$ to
tensor products of these $n$ functions.
 In all there
are $C_{n-1}$ star products, where $C_n$ is the Catalan
number of degree $n$. Of these there are $n-1$ triproducts which can serve as
physical products among closed string vertex operators. 
They are obtained from the bracketings
\beq 
\big((\cdots((\,\underline{f}\,_1\star \,\underline{f}\,_2)\star \,\underline{f}\,_3)\star\cdots )\star 
\,\underline{f}\,_r\big)\star\big(\,\underline{f}\,_{r+1}\star (\,\underline{f}\,_{r+2}\star(\,\underline{f}\,_{r+3}\star(\cdots \star \,\underline{f}\,_n)\cdots))\big)\label{ntripbra}
\eeq
for $1\leq r\leq n-1$, which for $r\leq n-2$ are obtained by successive applications for
$s=r, r+1,\dots, n-3$ of the inverse associators
\beq
\Phi^{-1}_{(\cdots((\alg\otimes\alg)\otimes
\alg)\otimes\cdots) \otimes\alg\,,\,\alg\,,\,\alg\otimes(\alg\otimes(\alg\otimes(\cdots
  \otimes\alg)\cdots )) }
\label{eq:Phij}\eeq
to $((\cdots
(\,\underline{f}{}_1\otimes\,\underline{f}{}_2)\otimes\,\underline{f}{}_3)\otimes
\cdots
) \otimes \,\underline{f}{}_r)\otimes \underline{f}{}_{r+1}\otimes (\underline{f}{}_{r+2}\otimes(\underline{f}{}_{r+3}\otimes(\cdots\otimes\underline{f}{}_n)\cdots)))$, where the first slot of $\Phi^{-1}$ in
(\ref{eq:Phij}) contains $s$ tensor products of the algebra
$\alg=C^\infty(\Mcal)$. The action of the inverse associator
(\ref{eq:Phij}) is again obtained using the Leibniz rule for the
partial derivatives $\partial_i$ on $\alg^{\otimes s}$ and
$\alg^{\otimes{n-s-1}}$. It follows that for the specific associator we
are considering there is no ambiguity in rewriting the expression
(\ref{eq:Phij}) simply as
$\Phi^{-1}_{\alg^{\otimes s}, \,\alg\,,\,\alg^{\otimes n-s-1\,}}$. Its
momentum space representation is obtained by applying it to the tensor products of plane waves
\bea
&&  \!\!\!\!\!\!\!\!\!\!\!\!\! \!\!\!\!\!\!\!\!\!\!\!\!\!\big[ \big( \cdots (\e^{\ii k_{1}\cdot x}\otimes \e^{\ii
    k_{2}\cdot x})  \otimes \cdots \big)
  \otimes \e^{\ii k_s\cdot x}\big]\ \otimes \
  \e^{\ii k_{s+1}\cdot x} \\ 
&& \qquad \qquad \qquad \qquad \qquad \qquad ~~\otimes \ 
\big[\e^{\ii k_{s+2}\cdot
  x}\otimes\big( \e^{\ii k_{b+3}\cdot x} \otimes(
  \cdots \otimes \e^{\ii k_{n}\cdot x} )\cdots \big)\big] \nn
\eea
which yields the phase factor
\beq\label{phasefact}
\exp\big( \mbox{$\frac{\ii \ell_s^4}6$}\, R(k_1+\cdots+
k_{s}\,,\, k_{s+1}\,,\, k_{s+2}+\cdots+ k_n)\big)~.
\eeq

The $n$-triproducts (\ref{ntripbra}) on phase space ${\cal M}$ induce triproducts on
configuration space $M$ by setting $\underline{f}{}_i=\pi^*f_i$ for
$i=1,\dots,n$ and then
 pulling these products back along the local sections $s_{\bar  p}:M\to
 \Mcal$.  Their explicit expressions in terms of the Moyal-Weyl
 product and of the $R$-flux tridifferential operators is obtained by
 replacing the tridifferential operators of Proposition~\ref{prop:ntriproduct} with
\eq
\e^{R_{(r,n)}}:=\exp\bigg( \frac{\ell_s^4}{12}\, \sum_{1\leq
    a<b<c\leq n} R^{ijk}\,  \partial_i^a\,\partial_j^b\, \partial_k^c
\,-\, \frac{\ell_s^4}{6}\ \sum_{b=r+1}^{n-1} \ \sum_{1\leq a<b<c\leq
      n}\,  R^{ijk}\, \partial_i^a\, \partial_j^{b}\, \partial_k^c\,
   \bigg)\ .
\ee
These $n$-triproducts also obey the on-shell condition of
Proposition~\ref{prop:onshell} (as can be seen from the
vanishing of the phase in (\ref{phasefact}) after integrating over
$M$ which yields momentum conservation $k_1+\cdots+k_n=0$).

The remaining triproducts violate the on-shell condition because they
are obtained from the previous ones by applying those inverse associators that like the
vertical arrows in (\ref{eq:pentagon}) do not act on all the tensor products entries;
in this case total momentum conservation does not imply
the trivialization of their action.
For example one has
\bea
&&\!\!\!\!\!\!\!\!\!\!\!\int_M\, \dd^dx\ 
s^*_{\bar p}\Big[
\big(((\cdots((\,\pi^*\!f_1\star \,\pi^*\!f_2)\star \,\pi^*\!f_3)\star\cdots )\star \pi^*\!f_{r-2})\star 
(\pi^*\!f_{r-1} \star \pi^*\!f_r)\big)\nn\\
&&\qquad \qquad \qquad \qquad \qquad \qquad \qquad \qquad ~~\star\big(\,\pi^*\!f_{r+1}\star (\,\pi^*\!f_{r+2}\star(\,\pi^*\!f_{r+3}\star(\cdots \star \,\pi^*\!f_n)\cdots))\big)
\Big]\nn\\[4pt]
&&\qquad \qquad \qquad=\int_M\, \dd^dx\ \As_{\bar
  p}\Big[\e^{R_{(r,n)}}\, \big(
\Phi^{-1}_{\alg^{\otimes {r-2}},\,\alg\,,\,\alg}\otimes\Id^{\otimes n-r}\big)(f_1\otimes\cdots \otimes
f_{n})\Big]\nn\\[4pt]
&&\qquad\qquad \qquad=\int_M\, \dd^dx\ 
\As_{\bar p}\Big[\Phi^{-1}_{\alg^{\otimes
    {r-2}},\,\alg\,,\,\alg}(f_1\otimes\cdots \otimes
f_{r})\Big]\,\As_{\bar p} \,f_{r+1}\,\As_{\bar p}\, f_{r+2}\,\cdots\,\As_{\bar p}\, f_{n} \label{eq350}~.
\eea
Because of momentum conservation, all $n$-triproducts which
differ from that above by a sequence of inverse associators
acting non-trivially on each tensor product entry (like the horizontal or
diagonal arrows in (\ref{eq:pentagon})) yield the same
integral (\ref{eq350}). 
This procedure can be iterated to obtain all remaining
triproducts and their on-shell associativity violating amplitudes, but
we refrain from detailing further the combinatorics.

\newsection{Triproducts from phase space cochains\label{sec:cochain}}

\subsection{Motivation: Differential geometry of nonassociative $R$-space}

In order to extend the considerations of Section~\ref{sec:star} to
more general geometric entities, such as differential forms and vector
fields, we need to uncover the more algebraic
construction that underlies the explicit representation via Fourier
transformation that we used so far. Such a description of
the phase space nonassociative geometry was introduced
by~\cite{Mylonas2013}, while in~\cite{Barnes2014}  a general
categorical perspective on cochain twist deformation is presented.
Here we shall further develop the phase space nonassociative geometry,
recover the main results of Section~\ref{sec:star} in this algebraic
context, and then study
the exterior product of forms in nonassociative phase space; we
further derive the
induced products on the
leaves of constant momentum and their property under integration. We also study the
deformed Lie algebra of infinitesimal diffeomorphisms on phase space,
as well as its action on
exterior forms and tensor fields. These geometric structures are again
induced on each constant momentum leaf and in particular in
configuration space (the leaf of vanishing momentum). We thus derive the
configuration space geometry of exterior forms and vector fields,
together with their deformed Lie algebra and Jacobiator, at all orders in the $R$-flux
components $R^{ijk}$ from the geometry of phase space that we canonically
construct via 2-cochain twist deformation.

\subsection{Cochain twist deformations and nonassociative star products}\label{sec:cohainprod}

Consider the Lie algebra $\mfh$ with generators
$P_i,\tP^i,M^i$ for $i=1,\dots, d$ and relations
\beq
[\tP^i,M^j]=\mbox{$\frac{1}{6}$} \, \ell_s^4\, R^{ijk}\, P_k~,\label{PMRP}
\eeq
with all other Lie brackets vanishing. 
This Lie algebra naturally acts on
functions on phase space via the representation
\eq
P_i( \,\underline{f}\, ) := \p_i\, \underline{f} ~,\qquad \tP^i (\, \underline{f}\, ) :=
\ii\hbar\, \tilde\p^i \, \underline{f} \qquad \mbox{and} \qquad  M^i
(\, \underline{f}\, ) =\frac{\ii\ell_s^4}{6 \hbar}\, R^{ijk}\, p_j\,\p_k\, \underline{f} \label{rep}
\ee
for all $\underline{f} \in C^\infty({\cal M})$.
The operators $P_i$ and $\tP^i$ respectively generate position and momentum translations in phase space, while $M^i$ generate Bopp shifts. Alternatively, we can represent $\mfh$ on the
algebra $\Diff(M)$ of differential operators on $M$ by
\beq
P_i= \ad_{\partial_i} \ , \qquad \tP^i= \ad_{x^i} \qquad \mbox{and} \qquad M^i= \mbox{$\frac{1}{6}$} \, \ell_s^4\, R^{ijk}\, \partial_j\, \ad_{\partial_k} \ ,
\label{repdiffM}\eeq
with $\ad_{\partial_i}(x^j)=\delta_i^j= -\ad_{x^j}(\partial_i)$.

Let $U(\mfh)$ be the universal enveloping algebra of the Lie algebra
$\mfh$, i.e., the associative algebra of $\complex$-linear combinations of
products of the generators $P_i, \tP^i, M^i$ modulo the Lie algebra relations
(\ref{PMRP}) as well as all other vanishing relations
$[P_i,P_j]=0, [P_i,M^j]=0, [M^i,M^j]=0$, etc.
This associative unital algebra is a Hopf algebra with coproduct $\Delta:
U(\mfh)\to U(\mfh)\otimes U(\mfh)$, counit $\varepsilon:  U(\mfh)\to
\complex$ and  antipode $S:U(\mfh)\to U(\mfh)$
defined on generators $X\in\{P_i,\tP^i,M^i\}$ as
\eq
\Delta(X)=X\otimes 1+1\otimes X~, \qquad \varepsilon(X)=0 \qquad \mbox{and} \qquad S(X)=-X \ .\label{DeS}
\ee
The coproduct $\Delta$ and counit $\varepsilon$ are linear maps extended multiplicatively to
all of $U(\mfh)$, while the antipode $S$ is a linear map extended
anti-multiplicatively to all of $U(\mfh)$, i.e., $S(X\,Y)=S(Y)\,
S(X)$. The action of $\mfh$ on $C^\infty({\cal M})$ naturally  extends to an
action of $U(\mfh)$ on $C^\infty({\cal M})$ via
\beq
(X_1\, X_2\cdots
X_n)(\, \uf\, ):=X_1\big(X_2(\cdots X_n(\, \uf\, )\cdots)\big)
\eeq
for all $X_1,X_2,\ldots, X_n\in \mfh$ and $\uf\in C^\infty({\cal M})$.

The coproduct on generators $X$ encodes the Leibniz
rule $X(\, \underline{f}\ \underline{g}\, )=X(\,\underline{f}\, )\,
\underline{g} +\underline{f}\, X(\, \underline{g}\,)$ for all
functions $\underline{f}\, ,\, \underline{g}$ on phase space. An
equivalent expression for the Leibniz rule is 
$X(\, \underline{f}\ \underline{g}\, )=\mu\circ\Delta(X)(\,\underline{f}\otimes
\underline{g}\,)$. Multiplicativity of the coproduct implies more
generally that 
\eq\label{leibD}
\zeta(\, \underline{f}\ \underline{g}\,)=\mu\big(\Delta(\zeta)(\, \underline{f}\otimes 
\underline{g}\, )\big)
\ee
for all $\zeta\in U(\mfh)$, which equivalently reads $\zeta\circ\mu\,(\,\underline{f}\ot 
\underline{g}\,)=\mu\circ\Delta(\zeta)(\,\underline{f}\otimes 
\underline{g}\,)$.

We can consider two twists in $U(\mfh)\otimes
U(\mfh)$ given by\footnote{The twists $F$ and $F'$ should be more precisely
  regarded respectively as power series expansions in $\hbar$ and
  $\ell_s^4/\hbar$, the deformation
  parameters that for ease of notation have been absorbed 
  into the definition of the generators $\tP^i$ and $M^i$ respectively. Then $U(\mfh)$ is an algebra
  over $\complex[[\hbar,\ell_s^4/\hbar]]$, the formal power series in
  $\hbar$ and $\ell_s^4/\hbar$ with
  coefficients in $\complex$. In this setting the twists $F$ and $F'$
  live in a completion of the tensor product $U(\mfh)\otimes U(\mfh)$.}
\eq
F=\exp\big(-\mbox{$\frac{1}{2}$}\, (P_i \ot \tP^i -\tP^i \ot P_i ) \big) \qquad \mbox{and} \qquad F' =\exp\big(-\mbox{$\frac{1}{2}$}\, (M^i \ot P_i -  P_i \ot M^i)\big) \ . \label{F}
\ee
They are both abelian cocycle twists of the Hopf algebra $U(\mfh)$ with the coalgebra structures $\Delta,
\varepsilon, S$, i.e., they are invertible and satisfy the relations
\eqa
(F\ot1)\, (\Delta\ot\Id)F &=& (1\ot F)\, (\Id\ot\Delta)F \ , \label{2cocycle}\\[4pt]
 (\eps\ot\Id)F &=& 1 \ = \ (\Id\ot\eps)F \ , \label{counital}
\eqaend
plus the analogous relations for $F\to F'$.
The first relation is a 2-cocycle condition which assures that the
star products obtained from the twists are associative. The
second relation is just a normalization condition.
Since the bracket (\ref{PMRP}) is central in the Lie algebra
$\mfh$, we can use the Baker--Campbell--Hausdorff formula 
\eq
\exp(A)\, \exp(B)=\exp\big([A,B]\big)\, \exp(B)\, \exp(A) \label{BCH}
\ee 
for $[A,B]$ central (and $A,B$ in $U(\mfh)\otimes U(\mfh)$ and then in
$U(\mfh)\otimes U(\mfh)\otimes U(\mfh)$) to compute
\eq
F\, F'=F'\, F \qquad \mbox{and} \qquad (1\otimes  F'\, )\, (\Id\otimes \Delta)F= \e^{R}\, (\Id\otimes \Delta) F\, (1\otimes F'\, )
\ee
where we have introduced the central element in $U(\mfh)\otimes
U(\mfh)\otimes U(\mfh)$ given by
\eq
\e^R:=\exp\big(\,\mbox{$\frac{\ell_s^4}{12}$}\, R^{ijk}\, P_i\otimes
P_j\otimes P_k\, \big) ~.
\ee
Note that, in the representation (\ref{rep}) on $C^\infty({\cal M})$,
the square of this operator is the associator (\ref{assoctridiff}), i.e.,
\eq
\Phi=\e^{2R}~.
\ee

The $n$-triproducts of functions from Section~\ref{sec:star} can also be obtained
from more algebraic considerations using the twists $F$ and $F'$ of the
Hopf algebra $U(\mfh)$. On one hand this requires a minimal amount of Hopf
algebra technology, while on the other hand this approach can be applied to
any algebra that carries a representation of $U(\mfh)$. In particular,
we shall apply it to the algebras of exterior differential forms and of Lie derivatives (infinitesimal diffeomorphisms).

Let us first consider the twist $F'$. 
Associated with $F'$ there is a new Hopf Algebra
$U(\mfh)^{F'}$. The algebra structure is the same as that of $U(\mfh)$,
the new coproduct is given by 
\eq 
\Delta^{F'}(\xi)=F'\, \Delta(\xi)\, F'^{-1}
\ee
for all $\xi\in U(\mfh)$, while $U(\mfh)^F$ has the same counit $\varepsilon$
and antipode $S$ as $U(\mfh)$ due to the abelian structure of the
twist, i.e., $P_i$ and $M^j$ commute.
We can now further deform the Hopf algebra $U(\mfh)^{F'}$ with the
twist $F$. Notice that while $F$ is a cocycle twist for the Hopf
algebra $U(\mfh)$, because it satisfies (\ref{2cocycle}) and
(\ref{counital}), 
it is \emph{not} a cocycle twist of the new Hopf algebra $U(\mfh)^{F'}$. To
compute its failure it is convenient to compare the actions of the
coproducts $\Delta^{F'}$ and $\Delta$ on the twist element $F$.
\begin{proposition}\label{pdf}
 \ $
\big(\Id\ot\Delta^{F'}\,\big)F=\e^{R}\,(\Id\ot\Delta)F \qquad \mbox{and}
\qquad 
\big(\Delta^{F'}\otimes \Id\big)F = \e^{-R}\, (\Delta\ot\Id)F \ . $
\end{proposition}
\begin{proof}
The coproduct $\Delta^{F'}$
 differs from $\Delta$ when applied to the generators $\tP^i$. We consider the term in
 (\ref{BCH}) which is linear in $B$, i.e., $\exp(A)\, B \,\exp(-A)=B+[A,B]$, and we
 identify $B$ with $\Delta(\tP^i)$ in order to easily compute
\eqa
\Delta^{F'} (\tP^i)&=&F'\,\Delta(\tP^i)\,F'{}^{-1} \nn\\[4pt] &=&
F'\,\big(\tP^i\otimes 1+1\otimes \tP^i\big)\, F'{}^{-1} 
\nn \\[4pt] &=& \tP^i\otimes 1+1\otimes \tP^i- \mbox{$\frac{1}{2}$}\,\big[M^k\ot
  P_k-P_k\ot M^k\,,\, \tP^i\otimes 1+1\otimes \tP^i\big]\nn\\[4pt]
&=&
\Delta(\tP^i) - \mbox{$\frac{1}{6}$} \,\ell_s^4\, R^{ijk}\, P_j\otimes P_k~.\label{DPt}
\eqaend
Next we compute
\eqa\label{idDF'F}
\big(\Id\ot\Delta^{F'}\,\big)F&=&\big(\Id\ot\Delta^{F'}\,\big)
\exp\big(-\mbox{$\frac{1}{2}$}\, (P_i \ot \tP^i -\tP^i \ot P_i )\big)\nn\\[4pt] &=&
\exp\Big(\mbox{$-\frac{1}{2}$}\, \big(P_i \ot \Delta^{F'}(\tP^i)- \tP^i \ot  \Delta^{F'}(P_i)\big) \Big)\nn \\[4pt]
&=&
\exp\Big(\mbox{$-\frac{1}{2}$}\, \big(P_i \ot \Delta(\tP^i) -\tP^i \ot  \Delta(P_i)
 \big) \Big) \,\exp\big(\mbox{$\frac{\ell_s^4}{12}$}\, R^{ijk}\, P_i\ot P_j\otimes P_k\big) \nn\\[4pt]
&=&(\Id\ot\Delta)F \,\e^{R} \ = \ \e^{R} \, (\Id\ot\Delta)F~,
\eqaend
where in the second equality we used multiplicativity of the coproduct
$\Delta$.
The proof of the second identity is very similar, or alternatively observe that $F\to F^{-1}$ under flipping of the
order of its legs in the tensor product $U(\mfh)\otimes
U(\mfh)$.
\end{proof}
As a corollary we find that $F$ is not a cocycle twist of $U(\mfh)^{F'}$, i.e., it
fails the cocycle condition for the Hopf algebra  $U(\mfh)^{F'}$.
\begin{corollary}\label{coroll}
 \ $(1\ot F)\, \big(\Id\ot\Delta^{F'}\, \big)F=\e^{2R}\, (F\ot1)\, \big(\Delta^{F'}\ot\Id\big)F ~.$
\end{corollary} 
\begin{proof} 
Use Proposition \ref{pdf} to rewrite the left-hand and right-hand
sides in terms of the undeformed coproduct $\Delta$. Then recall that $\e^R$
is central in the algebra $U(\mfh)\otimes U(\mfh)\otimes U(\mfh)$,
and finally use the cocycle property (\ref{2cocycle}).
\end{proof}

The Hopf algebra $U(\mfh)$ acts on the algebra
$C^\infty({\cal{M}})$ of functions on
phase space via the representation (\ref{rep}), 
and this action is compatible with
the product in  $C^\infty({\cal{M}})$ in the sense of (\ref{leibD}).
Because of this compatibility we can then deform $C^\infty({\cal{M}})$ into a new algebra
$C^\infty({\cal{M}})_{\star_{F'}}$ defined by the new product 
\eq
\underline{f}\, \star_{F'}\, 
\underline{g}=\mu_{F'}(\, \underline{f}\, \otimes \,
\underline{g}\,):= \mu\big({F'}{}^{-1}(\, \underline{f}\, \otimes \, \underline{g}\, ) \big)
\ee
for any phase space functions $\underline{f}$ and $\underline{g}\,$.
The algebra $C^\infty({\cal{M}})_{\star_{F'}}$ is noncommutative but
associative because of the cocycle condition satisfied by $F'$.
On $C^\infty({\cal{M}})_{\star_{F'}}$ there is a natural action of the
Hopf algebra  $U(\mfh)^{F'}$; it is again defined by (\ref{rep}). The
deformed coproduct $\Delta^{F'}$, which characterizes $U(\mfh)^{F'}$,
describes the action of any element $\zeta\in
U(\mfh)^{F'}$ on products of functions, i.e., for all $\underline{f}\,
, \, \underline{g}\in C^\infty({\cal{M}})_{\star_{F'}}$, by (\ref{leibD}) one has
\eqa
\zeta(\,\underline{f}\, \star_{F'}\, 
\underline{g}\, )&=&\zeta\circ \mu_{F'}(\, \underline{f}\, \otimes \,
\underline{g}) \nn \\[4pt] &=& \zeta\circ \mu\circ {F'}{}^{-1}(\,
\underline{f}\, \otimes \, \underline{g}\, ) \nn\\[4pt]
&=& \mu\circ \Delta(\zeta)\circ {F'}{}^{-1}(\, \underline{f}\,
\otimes \, \underline{g}\, ) \ = \ \mu_{F'}\circ\Delta^{F'}(\zeta) (\, \underline{f}\, \otimes \, \underline{g}\, ) ~.\label{leibD'}
\eqaend
In particular it leads to a deformed Leibniz rule for the generators
$\tilde P^i$ as can be read off from~(\ref{DPt}).

Now we iterate this procedure and deform the Hopf algebra
$U(\mfh)^{F'}$ and the noncommutative algebra $C^\infty({\cal{M}})_{\star_{F'}}$ 
with the 2-cochain $F\in U(\mfh)^{F'}\otimes U(\mfh)^{F'}$. 
 This is equivalent to the deformation of the original Hopf
 algebra $U(\mfh)$ and the algebra of smooth functions  $C^\infty({\cal{M}})$  with the 2-cochain
\eq
{\cal F}:=F\,F'=F'\,F~.
\ee
Since $F$ fails the $\Delta^{F'}$ 2-cocycle
property, the coproduct $(\Delta^{F'}\, )^F$, defined as 
\beq
(\Delta^{F'}\,)^F(\zeta)=F\, \Delta^{F'}(\zeta)\, F^{-1}
=\Delta^{\cal F}(\zeta) \ ,
\eeq
will equip
$(U(\mfh)^{F'}\,)^F=U(\mfh)^{\cal F}$ with the structure of a
quasi-Hopf algebra~\cite{Drinfeld} (see also~\cite{Mylonas2013,Barnes2014}). Correspondingly,
$C^\infty({\cal{M}})_{\star_{F'}}$ deforms to a quasi-associative noncommutative algebra with product
\eq
\underline{f}\, \star_{\cal F} \,\underline{g}
=\underline{f}\, \star_{F\, F'} \,\underline{g}=\mu_{F'}\big(F^{-1}(\,
\underline{f}\, \otimes \, \underline{g}\,
)\big)=\mu\big(({F\,F'}\,)^{-1}(\, \underline{f}\, \otimes \, \underline{g}\, ) \big)
\ee
for all phase space functions $\underline{f}\, ,\, \underline{g}\, $. 
The failure of associativity of $\star_{\cal{F}}$ is due to the
failure of the cocycle condition for $F$: We have
\eqa
(\, \underline{f}\, \star_{F\,F'}\, \underline{g}\, )\,
\star_{F\,F'}\, \underline{h}&=&\mu_{F\,F'}\circ (\mu_{F\,F'}\otimes\Id)(\,\underline{f}\, \otimes \,
\underline{g}\, \otimes \, \underline{h}\, ) \nn\\[4pt] &=&
\mu_{F'}\circ {F}^{-1}\circ (\mu_{F'}\otimes 1)\circ ({F}^{-1}\otimes \Id) (\, \underline{f}\, \otimes
\, \underline{g}\, \otimes \, \underline{h}\, )\nn\\[4pt]  &=&\label{eq424'}
\mu_{F'}\circ
(\mu_{F'}\otimes\Id)\circ(\Delta^{F'}\otimes\Id){F}^{-1}\, ({F}^{-1}\otimes 1) (\, \underline{f}\, \otimes
\, \underline{g}\, \otimes \, \underline{h}\, )\\[4pt]  &=&
\mu_{F'}\circ (\Id\otimes\mu_{F'}) \circ
(\Id\otimes\Delta^{F'}\, )F^{-1}\, (1\otimes F^{-1})\, \e^{2R} (\, \underline{f}\, \otimes
\, \underline{g}\, \otimes \, \underline{h}\, ) \nn 
\eqaend
where in the third line we used (\ref{leibD'}) on the first tensor
product entry of $F^{-1}$, while
in the fourth line we used associativity of the
$\star_{F'}$ product in the form $\mu_{F'}\circ (\mu_{F'}\otimes
\Id)=\mu_{F'}\circ (\Id\otimes \mu_{F'})$ together with the inverse of the
2-cochain property of Corollary~\ref{coroll}. Similarly, we have
\eqa\label{associat2}
\underline{f}\, \star_{F\,F'}(\, \underline{g} \, \star_{F\,F'}\,
\underline{h}\, )&=&\mu_{F\,F'}\circ(\Id\otimes \mu_{F\,F'}) (\, \underline{f}\, \otimes 
\, \underline{g}\, \otimes \, \underline{h}\, )~\nn\\[4pt]
&=&\mu_{F'}\circ(\Id\otimes \mu_{F'}) (\Id\ot \Delta^{F'}\, )F^{-1} \, (1\ot F^{-1})(\, \underline{f}\, \otimes 
\, \underline{g}\, \otimes \, \underline{h}\, ) \ .
\eqaend
Comparison of these products gives (\ref{algid329}).

For later use,  we note that we have chosen to derive the
nonassociativity structure of the $\star_\FF$ product using
associativity of the $\star_{F'}$ product and the 2-cochain property
of Corollary~\ref{coroll}. We could equally as well have decomposed the $\star_\FF$ product
as $\mu_{\FF}=\mu\circ\FF^{-1}$ and then, following similar steps as
in (\ref{eq424'}), used associativity of the
undeformed product $\mu$ and the 2-cochain property for $\FF=F\,F'$
(or for $\FF^{-1}$) which is not difficult to show to be given by
\eq 
(\Delta\otimes \Id)\FF^{-1}\, (\FF^{-1}\otimes 1)=(\Id\otimes 
\Delta)\FF^{-1}\, (1\otimes \FF^{-1})\, \e^{2R}~.\label{FFfundid}
\ee 

In a similar vein, using the isomorphism
$C^\infty(\Mcal)_{\star_F}\cong\Diff(M)$ of noncommutative associative
algebras and the representation (\ref{repdiffM}) of the Hopf algebra $U(\mfh)$, one can deform the algebra of differential operators on configuration space $M$ to a quasi-associative noncommutative algebra $\Diff(M)_{\star_{F'}}$.

\subsection{Configuration space triproducts}\label{sec:cst}

Let us now demonstrate how to reproduce the triproducts of Section~\ref{sec:star} within the present formalism. The pullback of the product on the leaf of constant
momentum $\bar{p}$ is given by
\eqa
s^\ast_{\bar{p}}(\, \underline{f}\, \star_{F\, F'}
\,\underline{g}\, )
&=&\mu_{F'}\big(F^{-1}(s_{\bar{p}}^*\, \underline{f} \, \otimes \,
s_{\bar{p}}^*\, \underline{g}\, ) \big) \nn\\[4pt]
&=& \mu_{F'}\big(s_{\bar{p}}^*\, \underline{f}\, \otimes \,
s_{\bar{p}}^*\, \underline{g}\, \big) \
= \ \mu\big(\exp\big(\, \mbox{$\frac{\ii\ell^4_s}{6\hbar}$}\,
\theta_{\bar{p}}\, \big)(s_{\bar{p}}^* \, \underline{f}\, \otimes \,
s_{\bar{p}}^*\, \underline{g}\, )\big)
\eqaend
where in the second equality we observed that
$F$ acts as the identity on functions of constant momentum, while in
the third equality we recalled the definition of the bivector
$\theta_{\bar{p}}$ from (\ref{defthetap}).
In particular if $f,g$ are functions on configuration space and
$\underline{f}=\pi^\ast f$,
$\underline{g}=\pi^\ast g$, then we recover the 2-product
$\mu^{(2)}_{\bar{p}}$ from (\ref{mu2bidif}),
\eq
s^\ast_{\bar{p}}\big(\pi^\ast f\, \star_{F\, F'}
\,\pi^\ast g \big)=f\, \As_{\bar{p}}\, g=\mu^{(2)}_{\bar{p}}(f\otimes g) ~.
\ee 

We now proceed to the product of three phase space functions and its pullback on the leaf of constant
momentum $\bar p$. We substitute in the triple product expression  (\ref{eq424'}) the inverse of
the second identity of Proposition~\ref{pdf} in order to express $\Delta^{F'}$ in terms of $\Delta$ and 
obtain
\eq \label{triprodff'}
(\underline{f}\,\star_{F\,F'}\,\underline{g}\,)\,\star_{F\,F'}
\,\underline{h}=\mu_{F'}\circ (\mu_{F'}\otimes\Id)\circ \e^R \,
(\Delta\otimes\Id){F}^{-1}\, ({F}^{-1}\otimes 1) (\, \underline{f}\, \otimes
\, \underline{g}\, \otimes \, \underline{h}\, )~.
\ee
The pullback of this expression reads as
\eqa \label{spullback}
s^\ast_{\bar{p}}\big((\,\underline{f}\,\star_{F\,F'}\,
\underline{g}\,)\star_{F\,F'} \,\underline{h}\,\big)
&=&
\mu_{F'}\circ(\mu_{F'}\otimes\Id)\circ \e^{R}\, (\Delta\otimes\Id)
F^{-1} \, (F^{-1}\otimes 1)
  \big(s^\ast_{\bar{p}}\, \underline{f}\,\otimes s^\ast_{\bar{p}}\,\underline{g}\,\otimes \,
  s^\ast_{\bar{p}}\,\underline{h}\,\big)\nonumber\\[4pt]
&=&
\As_{\bar{p}}\big(\e^{R}
  (s^\ast_{\bar{p}}\,\underline{f}\,\otimes s^\ast_{\bar{p}}\,\underline{g}\,\otimes \,
  s^\ast_{\bar{p}}\,\underline{h}\,)\big)
\eqaend
where we dropped both $F^{-1}$ and $(\Delta\otimes \Id)F^{-1}=\exp\big(\,
\frac{1}{2}\, (\Delta(P_i)\otimes \tilde
P^i-\Delta(\tilde{P}^i)\otimes P_i)\, \big)$ because  
they act trivially on functions
of constant momentum.
In particular if $f,g,h$ are functions on configuration space and
$\underline{f}=\pi^\ast f$,
$\underline{g}=\pi^\ast g$, $\underline{h}=\pi^\ast h$, then we recover the basic triproduct
$\mu^{(3)}_{\bar{p}}$ from~(\ref{basic3trip}), i.e.,
\eq\label{sast}
s^\ast_{\bar{p}}\big((\pi^\ast f \,\star_{F\,F'}\,
\pi^\ast g \,)\star_{F\,F'}\, \pi^\ast h\big) =
\As_{\bar{p}}\big(\e^{R}(f\otimes g\otimes h)\big)
=\mu^{(3)}_{\bar{p}}(f\otimes g\otimes h) ~.
\ee 

To generalise this computation to the pullback of the product
of $n$ 
functions we introduce the notation $F={F}^\alpha\otimes {F}{}_\alpha\in
U(\mfh)\otimes U(\mfh)$ (with summation over $\alpha$ understood) and 
define $F_{12}:=F\otimes 1$, 
$F_{23}:=1\otimes F$, $F_{13}:=F^\alpha\otimes 1\otimes  F_\alpha$ in $U(\mfh)^{\otimes 3}$. More generally $F_{ab}\in U(\mfh)^{\otimes n}$ is the element which is non-trivial only in the $a$-th and $b$-th factors of the tensor
product: $F_{ab} =1\otimes\cdots \otimes F^\alpha \otimes\cdots\otimes
F_\alpha \otimes \cdots\otimes 1$.
Similarly, in $U(\mfh)^{\otimes n}$ we set $R_{123}=R\otimes
1\otimes \cdots\otimes 1$ and more generally
\eq
R_{abc}=\mbox{$\frac{\ell_s^4}{12}$}\, R^{ijk}\, P^a_i\, 
P^b_j \, P_k^c \ ,
\ee
 where as before
$\zeta^a\in U(\mfh)^{\otimes n}$ for $\zeta\in U(\mfh)$ is the element which is non-trivial only in the $a$-th factor of the tensor
product: $\zeta^a=1\otimes\cdots \otimes \zeta \otimes\cdots\otimes 1$.

With these notations we have
\eqa\label{idDF-1}
(\Delta\otimes\Id)F^{-1}&=&
(\Delta\ot \Id)
\exp\big(\mbox{$\frac{1}{2}$}\, (P_i \ot \tP^i -\tP^i \ot P_i )\big)
\nn\\[4pt] &=&
\exp\Big(\mbox{$\frac{1}{2}$}\, \big(\Delta(P_i) \ot \tP^i-
\Delta(\tP^i) \ot P_i\big) \Big) \ = \ F^{-1}_{13}\, F^{-1}_{23} \ ,
\eqaend
where in the last step we expanded the coproduct and then observed that
the arguments of the exponential mutually commute.
Substituting this equality in the second identity of Proposition~\ref{pdf} we obtain
$\big(\Delta^{F'}\ot\Id\big)F^{-1}=\e^R\, F^{-1}_{13}\, F^{-1}_{23}$, which
is immediately generalized to
\eq\label{idDF}
\big(\Delta^{F'}\ot\Id^{\ot
  n-1}\big)F^{-1}_{1e}=\e^{R_{1\/2\,e+1}}\,F^{-1}_{1\,e+1}\, F^{-1}_{2\,e+1}~.
\ee
Using as in (\ref{idDF-1}) multiplicativity of the coproduct and
commutativity of the momentum algebra we also easily obtain
$\big(
\Delta^{F'}\ot \Id^{\otimes 2}\big)\e^R=\e^{R_{1\/3\/4}}\, \e^{R_{2\/3\/4}}$ which is immediately generalized to
\eq\label{idDR}
\big(\Delta^{F'}\ot \Id^{\ot n-1}\big)\e^{R_{1\/b\/c}}=\e^{R_{1\,
    b+1\, c+1}}\, \e^{R_{2\,b+1 \, c+1}} \ .
\ee
\begin{proposition}\label{*FF'n}
\eqa \nn
&&\!\!\!\!\!\!\!\!\!\!\!\!\!\!\!\big(\cdots(\,\underline{f}{}_1\star_{F\,F'} \underline{f}{}_2\,)\star_{F\,F'} \cdots 
\star_{F\,F'} \underline{f}{}_{n-1}\big)\star_{F\,F'}
\underline{f}{}_n \\ 
&& \qquad \qquad \qquad \qquad=~\mu_{F'}\circ(
\mu_{F'}\ot\Id)\circ\cdots \circ (\mu_{F'}\ot \Id^{\ot n-2}) \nn \\ &&
\qquad \qquad \qquad \qquad \qquad \qquad \qquad \circ\, 
\prod_{1\leq a<b<c\leq 
  n}\, \e^{R_{a\/b \/c}} \ \prod_{1\leq d<e\leq n}\,
F^{-1}_{d\/e}\big(\,\underline{f}{}_1\ot\cdots \otimes\,
\underline{f}{}_n\big) \ . \nn 
\eqaend 
\end{proposition}
\begin{proof}
The proof is by induction. The assertion holds for $n=3$ by
(\ref{triprodff'}). We suppose that it holds for $n>3$ and prove that it holds for
$n+1$ by computing
\eqa
&&\!\!\!\!\!\!\!\!\!\!
\big(\cdots (\,\underline{f}{}_1\star_{F\,F'} \underline{f}{}_2)\star_{F\,F'} \cdots 
\star_{F\, F'}
\underline{f}{}_{n}\big)\star_{F\,F'}\underline{f}{}_{n+1} \\
&& \quad = \ \mu_{F'}\circ(
\mu_{F'}\ot\Id)\circ\cdots\circ (\mu_{F'}\ot \Id^{\ot n-2}) \nn \\ &&
\quad \quad \circ\,
\prod_{1\leq a<b<c\leq 
  n}\, \e^{R_{a\/b \/c}} \ \prod_{1\leq d<e\leq 
  n}\, F^{-1}_{d\/e}
\big(\cdots(\,\underline{f}{}_1\star_{F\,F'} \underline{f}{}_2)\ot
\cdots \otimes\,
\underline{f}{}_{n}\big)\ot\,\underline{f}{}_{n+1}
\nn\\[4pt]
&&\quad = \ \mu_{F'}\circ(
\mu_{F'}\ot\Id)\circ\cdots\circ (\mu_{F'}\ot \Id^{\ot n-2}) \nn \\ &&
\quad \quad \circ\,
\prod_{1\leq a<b<c\leq 
  n}\, \e^{R_{a\/b \/c}} \ \prod_{1\leq d<e\leq 
  n}\, F^{-1}_{d\/e}\, (\mu_{F'}\ot \Id^{\ot
  n-1})F^{-1}_{1\/2}\big(\,\underline{f}{}_1\ot\cdots  \otimes\,
\underline{f}{}_n \otimes\, \underline{f}{}_{n+1}\big) \nn\\[4pt]
&&\quad = \ \mu_{F'}\circ(
\mu_{F'}\ot\Id)\circ\cdots\circ (\mu_{F'}\ot \Id^{\ot n-1}) \nn\\
&& \quad \quad 
\circ\, \bigg(\big(\Delta^{F'}\ot\Id^{\ot n-1} \big)\Big(\, 
\prod_{1\leq a<b<c\leq 
  n}\, \e^{R_{a\/b \/c}} \ \prod_{1\leq d<e\leq 
  n}\, F^{-1}_{d\/e}\;\Big)\bigg) \,
F^{-1}_{1\/2}\big(\,\underline{f}{}_1\ot\cdots \ot\, 
\underline{f}{}_n \otimes\, \underline{f}{}_{n+1}\big) \nn\\[4pt]
&&\quad = \ \mu_{F'}\circ(
\mu_{F'}\ot\Id)\circ\cdots \circ (\mu_{F'}\ot \Id^{\ot n-1}) \nn \\ && \quad
\quad \circ\, 
\prod_{1\leq a<b<c\leq 
  n+1}\, \e^{R_{a\/b \/c}} \ \prod_{1\leq d<e\leq n+1}\,
F^{-1}_{d\/e}\big(\,\underline{f}{}_1\ot\cdots \ot
\underline{f}{}_n\ot\, \underline{f}{}_{n+1}\big)\nn
\eqaend
where in the third step we used (\ref{leibD'}), and in the
last step we used (\ref{idDR}) to write
\eq
\big(\Delta^{F'}\ot \Id^{\ot n-1}\big) \,
\prod_{1\leq a<b<c\leq n}\, \e^{R_{a\/b \/c}}=
\prod_{3\leq b<c\leq n+1}\,\e^{R_{1\/b \/c}} \ 
\prod_{2\leq a<b<c\leq  n+1}\, \e^{R_{a\/b \/c}}
\ee
while from (\ref{idDF}) we obtain
\eqa
\big(\Delta^{F'}\ot \Id^{\otimes n-1}\big) \, \prod_{1\leq d<e\leq
  n}\, F^{-1}_{de}&=&
\prod_{3\leq e<  n+1}\, \e^{R_{1\/2 \,e}} \ 
\prod_{3\leq e\leq n+1}\, F^{-1}_{1\/e} \ 
\prod_{2\leq d<e\leq n+1}\, F^{-1}_{d\/e}\nn\\[4pt]
&=&
\prod_{3\leq e<  n+1}\, \e^{R_{1\/2 \, e}} \ 
\prod_{1\leq d<e\leq n+1}\, F^{-1}_{d\/e} \, F^{}_{12}
\eqaend
and the result follows.
\end{proof}
Proposition \ref{prop:ntriproduct} now follows as an easy corollary. In
Propositon~\ref{*FF'n}
the inverse twists $F_{d\/e}^{-1}$ have been ordered on the right and
therefore when we pullback this result along $s^\ast_{\bar{p}}$, as in (\ref{spullback})
we can drop these terms because their action is
trivial. We therefore obtain 
\eqa
&&\!\!\!\!\!\!\!\!\!\!s^\ast_{\bar{p}}\big[\big(\cdots
(\,\underline{f}{}_1\,\star_{F\,F'} \underline{f}{}_2)\star_{F\,F'} \cdots 
\star_{F\,F'} \underline{f}{}_{n-1}\,\big)\star_{F\,F'}
\underline{f}\, _n\big] \\
&&~~~~~~ \qquad =~\mu_{F'}\circ(
\mu_{F'}\ot\Id)\circ\cdots \circ(\mu_{F'}\ot \Id^{\ot n-2}) \circ
\prod_{1\leq a<b<c\leq 
  n}\, \e^{R_{a\/b \/c}}\big(\,\underline{f}{}_1\ot\cdots \ot \,
\underline{f}{}_n\big) \ , \nn 
\eqaend 
and Proposition~\ref{prop:ntriproduct} follows immediately by setting $\underline{f}{}_i=\pi^*f_i$.

\subsection{Differential forms and tensor fields}\label{DiffForms}

The algebraic techniques we have described in this section can be
applied to any algebra carrying a Hopf algebra symmetry. Hence we can
extend our results by considering the larger algebra of exterior
differential forms rather than just the algebra of functions. 
There is a natural extension of the representation (\ref{rep}) of  the
Lie algebra $\mfh$  on $C^\infty({\cal M})=\Omega^0(\Mcal)$ to the
vector space of exterior forms $\Omega^\bullet(\Mcal) =
C^\infty(\Mcal,\bigwedge^\bullet T^*\Mcal)
=\bigoplus_{n\geq0}\, \Omega^n(\Mcal)$; it is given by the Lie
derivative ${\cal L}$. This representation 
is extended to all of $U(\mfh)$ in the obvious way via $\Lie_{\xi\, \zeta}=\Lie_\xi\circ\Lie_\zeta$
for all $\xi,\zeta\in
U(\mfh)$; for example
$\Lie_{P_i\, P_j}=\Lie_{P_i}\circ \Lie_{P_j}$.
Furthermore, $\Omega^\bullet({\cal M})$ is an algebra with the associative exterior
product $\wedge$, and the Hopf algebra $U(\mfh)$ is a symmetry of this algebra
because of the Leibniz rule for the Lie derivative, or equivalently
because the exterior product is compatible with the coproduct of elements
of $U(\mfh)$ acting on forms (via the Lie derivative). For ease of notation we set
$\zeta(\, \underline{\eta}\, )=\Lie_\zeta(\,\underline{\eta}\,) $, for all $\zeta\in U(\mfh)$ and
$\underline{\eta}\,,\, \underline{\omega}\in \Omega^\bullet(\Mcal)$,
so that by (\ref{leibD}) the compatibility condition reads as
\eq
\zeta(\, \underline{\eta}\wedge \underline{\omega}\, )=
\wedge\big(\Delta(\zeta) (\, \underline{\eta}\, \ot \, \underline{\omega}\, )
\big)\label{compaw}~.
\ee
The algebra of exterior forms $\Omega^\bullet(\Mcal)$ can then be deformed
to the exterior algebra $\Omega^\bullet (\Mcal)_{\star_{F'}}$, as
vector spaces 
 $\Omega^\bullet(\Mcal)=\Omega^\bullet (\Mcal)_{\star_{F'}}$; the
 deformed
exterior product $\wedge_{\star_{F'}}$ in $\Omega^\bullet
(\Mcal)_{\star_{F'}} $ is given by
\eq
\underline{\eta}\wedge_{\star_{F'}}\underline{\omega}=\wedge\big(F'^{-1}(\,
\underline{\eta}\, \otimes\, \underline{\omega}\, ) \big)~,
\ee
for all $\underline{\eta}\,,\, \underline{\omega}\in
\Omega^\bullet(\Mcal)_{\star_{F'}}$.

Now all expressions in Sections~\ref{sec:cohainprod} and~\ref{sec:cst} 
 have been obtained by solely using:

\begin{itemize}
\item Hopf algebra properties of the twists $F$ and $F'$
 of Proposition~\ref{pdf}  and Corollary~\ref{coroll}.

\item Compatibility of the action of $U(\mfh)$ on $C^\infty({\cal
    M})$ with the coproduct in $U(\mfh)$ and the product
  in $C^\infty({\cal M})$ (and also in (\ref{eq424'})--(\ref{associat2}) the
  associativity of the product $\mu_{F'}$ in $C^\infty({\cal M})_{\star_{F'}}$).

\item Triviality of the action of momentum translations $\tilde{P}^i\in U(\mfh)$ on the images of the pullbacks
 $s^\ast_{\bar p}$ of sections $s_{\bar p}: M\to {\cal M}$ of constant
 momentum $\bar p$; this in particular implies that the twist $F$ acts as the
 identity on the image of $s^\ast_{\bar p}$.
\end{itemize}

Since the pullbacks $s^\ast_{\bar p}$ and $\pi^\ast$ naturally
extend to exterior forms, so that the three conditions above also hold true
for the algebra $\Omega^\bullet(\Mcal)$, we conclude that all
expressions in Section \ref{sec:cohainprod} and in Section~\ref{sec:cst} hold true also when we replace the product $\mu$ with
the exterior product $\wedge$ and functions $\underline{f}{}_i\in
C^\infty({\cal M})$ or $f_i\in C^\infty({M})$
with forms ${\underline{\eta}}{}_i\in \Omega^\bullet(\Mcal)$ or  
${\eta}_i\in \Omega^\bullet(M)$, and the action of
the universal enveloping algebra $U(\mfh)$ on forms is always via the Lie derivative.
In particular we have 
\eqa\label{eq443}
\wedge^{(n)}_{{\bar p}}(\eta_1\otimes\cdots\otimes\eta_n) &:=&
s^\ast_{{\bar p}}\big[\big(\cdots(\pi^* \eta_1 \,\wedge_{\star_{F\,F'}}\,\pi^*\eta_2)
\wedge_{\star_{F\,F'}}\, \cdots \big)\,\wedge_{\star_{F\,F'}}\,\pi^* \eta_n \big]\nn\\[4pt]
&=&
\wedge_{\As_{{\bar p}}}\circ 
\prod_{1\leq a<b<c\leq 
  n}\, \e^{R_{a\/b \/c}}\;
(\eta_1\otimes\cdots\otimes\eta_n) \ ,
\eqaend
for all $\eta_i \in \Omega^\bullet(M)$, $i=1,\dots, n$, where $\wedge_{\As_{\bar p}}$ is the Moyal--Weyl star product on forms
(cf. (\ref{eq:nmoyaldef}))
\eq
\wedge_{\As_{\bar p}}(\eta_1\otimes\cdots\otimes\eta_n)=
 \wedge\Big[\exp\Big(\,
\frac{\ii\ell_s^4}{6\hbar}\, \sum_{1\leq a<b\leq n}\, \theta_{\bar
    p}^{ij}\, \Lie_{\partial_i}^{\,a} \;\Lie_{\partial_j}^{\,b}\,
  \Big)(\eta_1\otimes\cdots\otimes \eta_n)\Big] ~.
\ee

In the general context of
integration of forms, we also find that the products $\wedge^{(n)}_{\bar p}$
under integration reduce to the Moyal--Weyl exterior product
$\wedge_{\As_{\bar p}}$.
\begin{proposition}
If $\eta_1,\dots,\eta_n\in \Omega^\bullet(M)$ are
forms on configuration space $M$ such that $\eta_1\wedge\cdots\wedge\eta_n$ is a top form in
$\Omega^\bullet(M)$, then 
\eq
\int_M\, \wedge_{\bar p }^{(n)}(\eta_1\otimes\cdots\otimes
\eta_n)=\int_M\, \eta_1\wedge_{\As_{\bar
  p}}\cdots \wedge_{\As_{\bar  p}}\eta_n~. \nn
\ee
\end{proposition}
\begin{proof}
  We first compute 
 \eqa \wedge^{(n)}_{\bar
    p}(\eta_1\otimes\ldots\otimes\eta_n)\!\!&=&
  \wedge_{\As{\bar p}} \circ \exp\Big(\,\sum_{1\leq
    a<b<c\leq n}\, R_{a\/b\/c}\, \Big)(\eta_1\ot\cdots
  \eta_n) \label{intetas} \\[4pt]
  &=&\eta_1\wedge_{\As_{\bar p}}\cdots\wedge_{\As_{\bar p}} \eta_n\,+\,\wedge_{\As_{\bar p}}
  \circ\sum_{1\leq a<b<c\leq n}\, R_{a\/b\/c}\circ {\cal{O}}\,
  (\eta_1\ot\cdots\ot
  \eta_n)\nn\\[4pt]
  &=&\eta_1\wedge_{\As_{\bar p}}\cdots\wedge_{\As_{\bar p}}
  \eta_n\nn \\ && +\,\frac{\ell_s^4}{12}\, \wedge_{\As{\bar p}} \circ
  R^{ijk}\, \sum_{1\leq
    a<b<c\leq
    n}\, \Lie^{\,a}_{\partial_i}\, \Lie^{\,b}_{\partial_j}\, \Lie^{\,c}_{\partial_k}\circ
  {\cal{O}}\, (\eta_1\ot\cdots\ot \eta_n)\nn\\[4pt]
  &=&\eta_1\wedge_{\As{\bar p}}\cdots\wedge_{\As_{\bar p}}\eta_n\nn\\
  && +\,\frac{\ell_s^4}{12}\, \wedge_{\As_{\bar p}} \circ
  R^{ijk}\, \sum_{b=2}^{n-1}\, \Lie^{\,b}_{\partial_j} \
  \sum_{a=1}^{b-1}\, \Lie^{\,a}_{\partial_i} \ \sum_{c=b+1}^n\,
  \Lie^{\,c}_{\partial_k} \circ
  {\cal{O}} \,(\eta_1\ot\cdots\ot \eta_n)\nn\\[4pt]
  &=&\eta_1\wedge_{\As_{\bar p}}\cdots\wedge_{\As_{\bar p}}\eta_n\nn
  \\ && +\,\frac{\ell_s^4}{12}\,\wedge_{\As_{\bar
      p}} \circ
  R^{ijk}\, \sum_{b=2}^{n-1} \ \sum_{e=1}^n\, \Lie^{\,e}_{\partial_j}
  \ \sum_{a=1}^{b-1}\, \Lie^{\,a}_{\partial_i} \ \sum_{c=b+1}^n\,
  \Lie^{\,c}_{\partial_k} \circ
  {\cal{O}} \,(\eta_1\ot \cdots\ot \eta_n)\nn\\[4pt]
  &=&\eta_1\wedge_{\As_{\bar
    p}}\cdots\wedge_{\As_{\bar p}}\eta_n\nn \\ &&
+\,\frac{\ell_s^4}{12}\, \Lie_{\partial_j}\circ\wedge_{\As_{\bar p}} \circ
  R^{ijk} \, \sum_{b=2}^{n-1} \ \sum_{a=1}^{b-1}\,
  \Lie^{\,a}_{\partial_i} \ \sum_{c=b+1}^n\, \Lie^{\,c}_{\partial_k} \circ
  {\cal{O}}\, (\eta_1\ot\cdots\ot \eta_n) \ . \nn
\eqaend
 In the second line we expanded the
  exponential by factoring the operator ${\sum_{a<b<c}\,  R_{abc}}$ as
 \bea
\exp\Big(\, \sum_{1\leq a<b<c\leq n}\, R_{abc}\,
\Big)&=&\Id+{\sum_{1\leq a<b<c\leq n}\, R_{abc}}+\frac{1}{2}\, \Big(\,
{\sum_{1\leq a<b<c\leq n}\, 
    R_{abc}}\, \Big)^2+\cdots \nn\\[4pt] &=:& \Id+{\sum_{1\leq a<b<c\leq n}\,
    R_{abc}}\circ {\cal O} \ .
\eea
  In the fifth line we used antisymmetry of $R^{ijk}$ to
  replace $\Lie^b_{\partial_j}$ with
  $\sum_{e=1}^n \, \Lie^e_{\partial_j}$.  In the last line we used
  the Leibniz rule 
\beq
\Lie_{\partial_i}(\eta_1\wedge_{\As_{\bar p}}\cdots\wedge_{\As_{\bar
    p}}\eta_n)=\wedge_{\As_{\bar p}}\circ\sum_{e=1}^n\,
\Lie^{\,e}_{\partial_j}(\eta_1\ot\cdots \ot \eta_n) \ .
\eeq
Next we use the Cartan formula for the Lie derivative in terms of the
exterior derivative and the contraction operator as
$\Lie_{\partial_j}=\imath_{\partial_j}\circ \dd+\dd\circ \imath_{\partial_j}$,
observe that when acting on a  top form it simplifies to
$\Lie_{\partial_j}=\dd\circ \imath_{\partial_j}$, and the result then follows by 
integrating (\ref{intetas}). 
\end{proof}

Similarly to the exterior algebra one can also deform the tensor algebra.
As for the deformed exterior product $\wedge_{\star_\FF}$, the deformed tensor product
$\otimes_{C^\infty({\cal M})_{\star_\FF}}$ is defined by composing the usual tensor
product $\otimes_{C^\infty({\cal M})}$ over $C^\infty({\cal M})$ with
  the inverse twist: 
$\otimes_{C^\infty({\cal M})_{\star_\FF}}:=\otimes_{C^\infty({\cal M})}\circ \FF^{-1}$.

\subsection{Phase space diffeomorphisms}

The Drinfeld twist deformation procedure we have been implementing
consists in deforming algebras that carry a compatible representation
of the Hopf algebra $U(\mfh)$ 
(cf. the first two items in the list of  Section \ref{DiffForms}). We
recall that the representation is compatible with the product in the
algebra if the action of
Lie algebra elements $X\in\mfh$ on products is given by the Leibniz
rule (and hence for elements $\zeta\in U(\mfh)$ by the
coproduct). We shall now apply this procedure to the Lie algebra of
vector fields $\vect(\Mcal)=C^\infty(\Mcal, T\Mcal)$ on phase space, i.e., the Lie algebra of
infinitesimal (local) diffeomorphisms. This is a nonassociative algebra
with product given by the Lie bracket $[\quad]:\vect(\Mcal)\otimes \vect(\Mcal)\to \vect(\Mcal)$.
Thanks to the representation (\ref{rep}) the Lie algebra $\mfh$ can be
regarded as a subalgebra of $\vect(\Mcal)$ and therefore its action is given by the Lie
derivative:
$\Lie_X(\,\uu\,)=[X,\uu\,]$ for all $X\in \mfh$, $\uu\in \vect(\Mcal)$. The compatibility condition is satisfied because
the Jacobi identity
$[X,[\,\uu\,,\,\uv\,]]=[[X,\uu\,],\uv\,]+[\,\uu\,,\,[X,\uv\, ]]$ is the Leibniz
rule with respect to the product $[\quad]$.

We can thus apply the Drinfeld twist deformation procedure with
2-cochain ${\cal F}=F\, F'$ and obtain the deformed algebra of vector
fields $\vect(\Mcal)_{\star_{\cal F}}=\vect(\Mcal)_{\star_{F\, F'}}$, which
as a vector space is the same as $\vect(\Mcal)$ but has a deformed Lie bracket 
\begin{eqnarray}
[\quad]_{\star_{\cal F}}\,:\, \vect(\Mcal)_{\star_{\cal F}}\otimes
  \vect(\Mcal)_{\star_{\cal F}}&\longrightarrow& \vect(\Mcal)_{\star_{\cal F}} \, \nonumber\\
\uu\otimes\uv &\longmapsto& [\,\uu\,,\,\uv\,]_{\star_{\cal F}}:=
[\quad]\circ {(F\,F'\,)}^{-1}(\,\uu\otimes \uv\, ) \ .
\end{eqnarray}
This can be realized as a deformed commutator. For this, we introduce the
notation $\FF=\FF^\al\otimes\FF_\al$, 
$\FF^{-1}=\overline{\FF}{}^{\al}\otimes\overline{\FF}_\al$;
define the \emph{universal ${\cal R}$-matrix} ${\cal R}=\FF_{21}\, \FF^{-1}$  where
$\FF_{21}=\FF_\al\otimes \FF^\al$ and denote its inverse by ${\cal
  R}^{-1}=\overline{\cal R}{}^\al\otimes \overline{\cal R}_\al$. We then compute
\eqa 
[\,\uu\,,\,\uv\,]_{\star_{\FF}}&=&\big[\,
\overline{\FF}{}^\alpha(\,\uu\,),\overline{\FF}_\alpha(\, \uv\, ) \, \big] \nn
\\[4pt] &=& \overline{\FF}{}^\alpha(\,\uu\,)\,
\overline{\FF}_\alpha(\,\uv\, )-\overline{\FF}_\alpha(\, \uv\, )\, \overline{\FF}{}
^\alpha(\,\uu\, ) \ = \ \uu \, \star_{\FF}\, \uv-\overline{\cal
  R}{}^\alpha(\,\uv\, )
 \star_{\FF}\overline{\cal R}_\alpha(\,\uu\,)~,
\eqaend
where we wrote the undeformed Lie bracket as a
commutator and introduced the $\star_{\FF}$ product between vector fields
$\uu\, \star_{\FF} \, \uv:=\overline{\FF}{}^\al(\,\uu\,)\,
\overline{\FF}_\al(\,\uv\, )$ which is a
deformation of the product on the universal enveloping algebra of
$\vect(\Mcal)$.
It is easy to see that the bracket  $[\quad]_{\star_{\FF}} $ 
has the $\star_{\FF}$ antisymmetry property
\beq\label{sigmaantysymme}
[\,\uu\,,\,\uv\,]_{\star_\FF} =\big[\,
\overline{\FF}{}^\al(\,\uu\,),\overline{\FF}_\al(\,\uv\,)\, \big]=-\big[\,
\overline{\FF}_\al(\,\uv\,),
\overline{\FF}{}^\al(\,\uu\,)\, \big]=
-\big[\,\overline{\cal R}{}^\al(\,\uv\,), \overline{\cal
  R}_\al(\,\uu\,)\, \big]_{\star_\FF} \ .
\eeq
We can write this relation as
$[\quad]\circ\FF^{-1}(\,\uu\otimes \uv\, )=-[\quad]\circ\FF^{-1}\,{\cal R}^{-1}\circ \sigma
(\,\uu\otimes \uv\,)$ where $\sigma$ is the transposition $\sigma
(\,\uu\otimes \uv\,)=\uv\otimes \uu$. This notation emphasizes the
antisymmetry of the
$\star_\FF$ bracket with respect to transpositions implemented by the operator ${\cal R}^{-1}\circ\sigma$.

We can next implement  the transposition of the elements $\uu$ and $\uv$ in
the triple tensor product
$\uu\otimes (\uv\otimes \uz)\in  \vect(\Mcal)_{\star_\FF}\otimes (\vect(\Mcal)_{\star_\FF}\otimes
\vect(\Mcal)_{\star_\FF})$. The bracketing hints that we are later on going
  to apply the operator $\Id\otimes [\quad ]$. There is actually deeper information
  in the bracketing~\cite{Drinfeld, Barnes2014}, as it defines the
  action of the ${\cal F}$ deformed Hopf algebra of
  translations and Bopp shifts on triple tensor products. We therefore first have to identify $\vect(\Mcal)_{\star_\FF}\otimes (\vect(\Mcal)_{\star_\FF}\otimes
\vect(\Mcal)_{\star_\FF})$ with $(\vect(\Mcal)_{\star_\FF}\otimes \vect(\Mcal)_{\star_\FF})\otimes
\vect(\Mcal)_{\star_\FF}$ as quasi-Hopf algebra representations, and this is done via
the action of the inverse associator $\Phi^{-1}=\e^{-2R}$; 
then we can apply the operator ${\cal R}_{12}^{-1}\circ\sigma_{12}
=({\cal R}^{-1}\circ\sigma)\otimes \Id$ and finally go back to
$\vect(\Mcal)_{\star_\FF}\otimes (\vect(\Mcal)_{\star_\FF}\otimes \vect(\Mcal)_{\star_\FF})$ via
$\e^{2R}$. Thus the transposition map in this case is
$\e^{2R}\circ {\cal R}^{-1}_{12}\circ \sigma_{12}\circ \e^{-2R}$ which,
since $\e^R$ is central and $R_{123}=-R_{213}$, simplifies to 
${\cal R}^{-1}_{12}\circ \sigma_{12}\circ \e^{-4R}$. Hence by setting
(with summation over the indices ${}^{\bar\phi\bar\phi}$ understood)
\eq\label{betauvz}
\uu^{\pphim}\otimes  \uv^{\pphim}\otimes \uz^{\pphim}:=\e^{-4R} (\,
\uu\otimes \uv\otimes \uz\, ) \ ,
\ee
it explicitly acts as
\bea
\e^{2R}\circ {\cal R}^{-1}_{12}\circ \sigma_{12}\circ
\e^{-2R}(\, \uu\otimes \uv\otimes \uz\, )&=&
 {\cal R}^{-1}_{12}\circ \sigma_{12}\circ \e^{-4R}(\, \uu\otimes \uv\otimes
 \uz\, )\nn \\[4pt] &=&
\overline{\cal R}{}{}^{\al}(\, \uv^{\pphim}\, )\otimes \overline{\cal
  R}_\al(\, \uu^{\pphim}\, )\otimes \uz^{\pphim}~.\label{tuvhopf}
\eea
Independently of its quasi-Hopf algebraic aspects the relevance of
the expression (\ref{tuvhopf})
is in its appearence in the $\star_\FF$ Jacobi identity.
\begin{proposition}\label{FFJacobi} The deformed Lie algebra of infinitesimal
diffeomorphisms is characterized by the $\star_\FF$ Jacobi identity
\eq
\big[\,\uu \,,\,[\,\uv\,,\uz\,]_{\star_\FF} \, \big]_{\star_\FF} =\big[[\,
\uu^{\bar\phi},\uv^{\bar\phi}\, ]_{\star_\FF}  \,,\,\uz^{\bar\phi}\, \big]_{\star_\FF}  
+ \big[\, \overline{\cal R}{}^\al(\, \uv^{\pphim}\, )\,,\, \big[\,
\overline{\cal R}_\al(\, \uu^{\pphim}\, ),\uz^{\pphim}\, \big]_{\star_\FF} \big]_{\star_\FF} ~,
\nn \ee
for all $\uu\, ,\, \uv\, ,\, \uz\, \in\vect(\Mcal)_{\star_{\FF}}$, where 
\eq
\uu^{\bar\phi}\otimes \uv^{\bar\phi}\otimes \uz^{\bar\phi}:=
\Phi^{-1}(\, \uu\otimes \uv\otimes \uz\, )=
\e^{-2R}(\, \uu\otimes \uv\otimes \uz\, )\nn
\ee
 and 
$\uu^{\pphim}\otimes \uv^{\pphim}\otimes \uz^{\pphim}$ is
defined in (\ref{betauvz}).
\end{proposition}
\begin{proof}
We compute
\eqa\label{p1}
\big[\,\uu \,,\, [\,\uv\,,\uz\, ]_{\star_\FF} \, \big]_{\star_\FF} &=&
[\quad]\circ\FF^{-1}\circ (\Id \otimes
[\quad])\circ (1\otimes\FF^{-1})(\, \uu\otimes \uv\otimes \uz\, )
\nn\\[4pt]
&=&
[\quad]\circ(\Id \otimes
[\quad])\circ (\Id \otimes\Delta
)\FF^{-1}\, (1\otimes\FF^{-1})(\, \uu\otimes \uv\otimes \uz\, )
\eqaend
and
\eqa\label{p2}
\big[[\, \uu^{\bar\phi},\uv^{\bar\phi}\, ]_{\star_\FF} \,,\,
\uz^{\bar\phi}\, \big]_{\star_\FF}
&=&[\quad]\circ([\quad]\otimes\Id )\circ(\Delta\otimes \Id )\FF^{-1}\,
(\FF^{-1}\otimes 1)(\, \uu^{\bar\phi}\otimes \uv^{\bar\phi}\otimes
\uz^{\bar\phi}\, )\nn\\[4pt]
&=&[\quad]\circ([\quad]\otimes\Id )\circ
(\Id\otimes\Delta)\FF^{-1}\, (1\otimes\FF^{-1}) (\, \uu\otimes
\uv\otimes \uz\, )
\eqaend
where we used the 2-cochain property (\ref{FFfundid}). It follows that
\eqa\label{p3}
&& \big[\, \overline{\cal R}{}^\al(\, \uv^{\pphim})\,,\, [\, \overline{\cal
   R}_\al(\,\uu^{\pphim}\, ), \uz^{\pphim}\, ]_{\star_\FF}  \, \big]_{\star_\FF}
\nn \\ && \qquad \qquad \ = \ [\quad]\circ(\Id \otimes
[\quad])\circ (\Id \otimes\Delta )\FF^{-1}\, (1\otimes\FF^{-1}) \, 
 {\cal R}^{-1}_{12}\circ \sigma_{12}\circ \e^{-4R}
(\, \uu\otimes \uv\otimes \uz \, )\nn\\[4pt]
&& \qquad \qquad \ = \
[\quad]\circ(\Id \otimes
[\quad])\circ
(\Delta\otimes\Id)\FF^{-1}\, (\FF\otimes
1)\, \e^{-2R}\circ \sigma_{12}\circ \e^{-4R}
(\,\uu\otimes \uv\otimes \uz\, )\nn\\[4pt]
&& \qquad \qquad \ = \ [\quad]\circ(\Id \otimes
[\quad])\circ
\sigma_{12}\circ
(\Delta\otimes\Id)\FF^{-1}\, (\FF^{-1}\otimes
1)\, \e^{-2R}
(\, \uu\otimes \uv\otimes \uz\, )\nn\\[4pt]
&& \qquad \qquad \ = \
[\quad]\circ(\Id \otimes
[\quad])\circ
\sigma_{12}\circ
(\Id\otimes \Delta)\FF^{-1}\, (1\otimes \FF^{-1})
(\, \uu\otimes\uv\otimes \uz \, )
\eqaend
where in the third line we used the 2-cochain property
(\ref{FFfundid}) and noticed that since in our case
$\FF_{21}=\FF^{-1}$ we have  ${\cal
  R}^{-1}=\FF\, \FF^{-1}_{21}=\FF^2$. In the fourth line we moved the
permutation $\sigma_{12}$ to the left and used $\e^{-2R}\circ
\sigma_{12}=\sigma_{12}\circ \e^{2R}$, $\FF\circ\sigma=\sigma\circ
\FF^{-1}$ and the fact that the undeformed coproduct $\Delta(\zeta)$ of any
element $\zeta\in U({\mfh})$ is {\it cocommutative}, i.e., it is
symmetric with respect to exchange of its two legs, see (\ref{DeS}).
In the last line we again used (\ref{FFfundid}).
The proof now follows by observing that the undeformed Jacobi identity is
equivalent to the equality
\beq
[\quad]\circ(\Id \otimes
[\quad])=[\quad]\circ([\quad]\otimes\Id) +[\quad]\circ(\Id \otimes
[\quad])\circ
\sigma_{12}
\eeq
of operators on tensor products of vector fields.
\end{proof}

Note that the underlying mathematical structure we are
constructing is that of a quasi-Lie algebra of a quasi-Hopf algebra.  
The original Hopf algebra is the universal enveloping algebra of vector
fields $U(\vect(\Mcal))$. This is deformed to a quasi-Hopf algebra
$U(\vect(\Mcal))^\FF$ via the 2-cochain $\FF\in U(\mfh)\otimes U(\mfh)\subset
U(\vect(\Mcal))\ot U(\vect(\Mcal))$ (with the inclusion via the representation (\ref{rep})). The quasi-Lie algebra
of vector fields $\vect(\Mcal)_{\star_\FF}$ is the linear space $\vect(\Mcal)$ with the bracket
$[\quad]_{\star_\FF}  : \vect(\Mcal)_{\star_\FF}\otimes \vect(\Mcal)_ {\star_\FF}\to \vect(\Mcal)_{\star_\FF}$.

In the commutative case the commutator $[\,\uu\,,\uv\,]$ is equal to the
Lie derivative $\Lie_\uu(\, \uv\, )$.
This motivates the definition of the $\star_\FF$ Lie derivative as
\eq
\Lie^{\star_\FF}_\uu:=\Lie_{\overline{\FF}{}^\al(\,\uu\,)}\circ
\overline{\FF}_\al~,\label{defstLieder}
\ee
for
all $\uu\,,\,\uv\in \vect(\Mcal)_{\star_\FF}$, so that $\Lie_\uu^{\star
  \FF}(\,\uv\, )=[\, \uu\, ,\uv\, ]_{\star_\FF}$. 
The $\star_\FF$ Lie derivative is given by composing the usual
Lie derivative with the inverse twist $\FF^{-1}$.
The definition (\ref{defstLieder}) holds more generally when the
$\star_\FF$ Lie derivative acts on exterior forms or tensor fields.
The same algebra as that of Proposition~\ref{FFJacobi} shows that the
$\star_\FF$ Lie derivative satisfies the deformed Leibniz rule
\eq\label{Leiblie}
\Lie^{\star_\FF}_\uu(\, \underline\eta\wedge_{\star_\FF}
\underline\omega\, )=\Lie_{\uu^{\bar\phi}}^{\star_\FF}(\,
\underline\eta^{{\bar\phi}}\, )\wedge_{\star_\FF} \underline\omega^{\bar\phi} + 
\overline{\cal R}{}^\al(\, \underline\eta^{\pphim}\,
)\wedge_{\star_\FF} \Lie_{\overline{\cal R}_\al(\, \uu^{\pphim}\,
  )}^{\star_\FF}(\, \underline\omega^{\pphim}\, ) \ ,
\ee 
for
all $\uu\in \vect(\Mcal)_{\star_\FF}$,
$\underline{\eta}\,,\,\underline{\omega}\in\Omega^\bullet
(\Mcal)_{\star_\FF}$, where 
$\uu^{\bar\phi}\otimes {\underline\eta}^{\bar\phi}\otimes
\underline\omega^{\bar\phi}= \e^{-2R}(\, \uu\otimes
\underline\eta\otimes \underline\omega\, )$ 
 and 
$\uu^{\pphim}\otimes \underline\eta^{\pphim}\otimes\underline\omega^{\pphim}=
\e^{-4R}(\, \uu\otimes \underline\eta\otimes \underline\omega\,
)$. For this, we note that the Leibniz rule for the undeformed Lie derivative can be 
written as $\Lie\circ (\Id\otimes \wedge)=\wedge\circ
(\Lie\otimes\Id)+\wedge\circ(\Id\otimes\Lie)\circ\sigma_{12}$, where
we used the notation $\Lie(\, \uu\,,\,\underline\eta\, ):=\Lie_\uu(\,
\underline\eta\, )$. The Leibniz rule
for tensor fields is then obtained by replacing differential forms with tensor fields and
the deformed exterior product $\wedge_{\star_\FF}$ with the deformed
tensor product $\otimes_{{C^\infty\!({\cal M})}_{\star_\FF}}$.

\subsection{Configuration space diffeomorphisms}

We have described the deformed  Lie algebra of infinitesimal
diffeomorphisms on noncommutative and nonassociative phase space, and their action on phase space forms and tensors.
We next study the induced deformed Lie algebra and action in
configuration space. For this, we take advantage of the global coordinate system $\{x^i,p_i\}$
on phase space to define pullbacks of vector fields (more generally we could consider any phase space $\cal
M$ that is foliated with leaves isomorphic to configuration space
$M$).  Any vector field on ${\cal M}$ is of the form
$\uv=v^i(x,p)\, \partial_i+\tilde{v}_i(x,p)\, \tilde{\partial}^i$.
We define the Lie subalgebra of vector fields
\eq
{}^H\vect(\Mcal)=\big\{\uv\in\vect(\Mcal)~\big|~\uv=v^i(x,p)\, \partial_i
\big\}
\ee
and observe that the Lie algebra $\mfh$ leaves ${}^H\vect(\Mcal)$ invariant:
$\big[\mfh,{}^H\vect(\Mcal) \big]\subset {}^H\vect(\Mcal)$, hence we can
apply the deformation procedure to ${}^H\vect(\Mcal)$ and obtain the deformed
Lie algebra ${}^H\vect(\Mcal)_{\star_\FF}$. 

There is an obvious one-to-one correspondence between coordinate
vector fields $\partial_i$ on configuration space $M$ and
coordinate vector fields $\partial_i$ on phase space ${\cal M}$. 
Thus, on one hand we can now inject the Lie algebra $\vect(M)$ in ${}^H\vect(\Mcal)$ by defining
\eq
\pi^\ast: \vect(M)\longrightarrow {}^H\vect(\Mcal) \qquad \mbox{with}
\quad v^i\, \partial_i\longmapsto \pi^\ast(v^i)\, \partial_i~.
\ee
On the other hand, given the section $s_{\bar{p}}:M\to{\cal M}~,~
x\mapsto (x,\bar{p})$ we can project the Lie algebra ${}^H\vect(\Mcal)$ to $\vect(M)$ by defining
\eq
s_{\bar p}^\ast: {}^H\vect(\Mcal)\longrightarrow \vect(M) \qquad
\mbox{with} \quad \uv=v^i\, \partial_i\longmapsto s^*_{\bar p}(v^i)\, \partial_i~.
\ee
 We can then characterize the deformed Lie algebra of infinitesimal
diffeomorphisms $\vect(M)_{\bar{p}}$ as the vector space $\vect(M)$ with
the brackets 
\bea
[u,v]^{(2)}_{\bar{p}}&:=&s^*_{\bar{p}}\big([\pi^*u,\pi^*v]_{\star_\FF}\big)
\ , \\[4pt] \label{threer}
\big[u\,,\,[v,z]\big]^{(3)}_{\bar{p}}&:=&
s^*_{\bar{p}}\big(\big[\pi^*u\,,\,[\pi^*v,\pi^*z ]_{\star_\FF}
\big]_{\star_\FF}\big) \ , \\[4pt] \label{threel}
\big[[u,v]\,,\,z \big]^{(3)}_{\bar{p}}&:=& s^*_{\bar{p}}\big(\big[[\pi^*u,\pi^*v]_{\star_\FF} \,,\,\pi^*z]_{\star_\FF}\big)~.
\eea

Since no derivatives in the momentum
directions appear, for two vector fields we find explicitly
\eq\label{lievect}
[u,v]^{(2)}_{\bar{p}}=s^*_{\bar{p}}\circ [\quad]\circ
{F'}{}^{-1}\, F^{-1}(\pi^* u\otimes \pi^* v)=s^*_{\bar{p}}\circ
[\quad]\circ {F'}{}^{-1}\big(\pi^* u\otimes \pi^* v\big)=[u,v]_{\As_{\bar{p}}}
\ee
where the bracket 
\eq\label{qLie}
[u,v]_{\As_{\bar{p}}}:=
[\quad]\circ 
\exp\big(\mbox{$\frac{\ii\ell_s^4}{6\hbar}\, \theta_{\bar
    p}^{ij}\, \Lie_{\partial_i}\otimes \Lie_{\partial_j}$}\big) (u\otimes v)
\ee
is the quantum Lie algebra bracket on Moyal--Weyl noncommutative
space~\cite{GR2} (see~\cite[Section~8.2.3]{book} for an
elementary introduction). In terms of the twist $F_{{\bar p}}$
and its universal ${\cal R}$-matrix defined by
\eqa
F_{{\bar p}}=\exp\big(\mbox{$-\frac{\ii\ell_s^4}{6\hbar}\, \theta_{\bar
    p}^{ij}\, \Lie_{\partial_i}\otimes \Lie_{\partial_j}$}\big) \qquad &\mbox{and}& \qquad
F^{-1}_{{\bar p}}=
\overline{F}{}^\al_{{\bar p}}\otimes {\overline{F}_{{\bar p \:\al}}}=
\exp\big(\mbox{$\frac{\ii\ell_s^4}{6\hbar}\, \theta_{\bar
    p}^{ij}\, \Lie_{\partial_i}\otimes \Lie_{\partial_j}$}\big) \ , \nn\\[4pt]
R_{{\bar  p}}={F_{{\bar p}}}_{\,21}\, F^{-1}_{{\bar
    p}}=F^{-2}_{{\bar p}} \qquad &\mbox{and}& \qquad R^{-1}_{{\bar 
    p}}=
\overline{R}{}^\al_{{\bar p}}\ot {\overline{R}_{{\bar p
      \,\al}}}=F^{2}_{{\bar p}} \ , 
\label{eq469}\eqaend
we have $[\quad]_{\As_{\bar{p}}}=[\quad]\circ F^{-1}_{{\bar p}}$
  and the $\As_{\bar p}$ antisymmetry property
\beq
[u,v]_{\As_{\bar{p}}}=-\big[\,\overline{R}^\al_{{\bar
      p}}(v),{\overline{R}_{{\bar p \, \al}}}(u)\,
  \big]_{\As_{\bar{p}}} \ .
\eeq
The bracket (\ref{lievect}) also defines the Lie derivative on vector fields.

The explicit expression for the 3-bracket
$[[u,v],z]^{(3)}_{\bar{p}}$ can be obtained by following the same
steps as in (\ref{triprodff'})--(\ref{sast}). We just substitute $\mu_{F'}$ with
$[\quad]_{F'}:=[\quad]\circ F'{}^{-1}$ and functions with vector fields; we proceed similarly
with the three bracket $[u,[v,z]]^{(3)}_{\bar{p}}$ (cf. (\ref{associat2})) and
obtain 
\bea
\big[[u,v]\,,\,z\big]^{(3)}_{\bar{p}}&=&
[\quad]_{\As_{\bar{p}}} \circ 
\big([\quad]_{\As_{\bar{p}}}\otimes\Id \big) \circ  \e^{R}
(u\otimes v \otimes z)~, \nn\\[4pt]
\big[u\,,\,[v,z]\big]^{(3)}_{\bar{p}}&=&[\quad]_{\As_{\bar{p}}} \circ\big( \Id\otimes
[\quad]_{\As_{\bar{p}}} \big) \circ  \e^{-R} (u\otimes v\otimes z)~.
\eea
The relation between the 3-brackets
$[\quad[\quad]]^{(3)}_{\bar p}$ and $[[\quad]\quad]^{(3)}_{\bar p}$
can be obtained from these explicit expressions: From the $\As_{\bar
  p}$ antisymmetry property of the bracket $[\quad]_{\As_{\bar{p}}}$,
the easily checked property of the universal
${\cal R}$-matrix $\big(\Id\ot\Delta^{F_{\bar p}} \big)R^{-1}_{\bar p}=R^{-1}_{\bar
  p\;12}\, R^{-1}_{\bar p\;13}$  and the identity $\zeta\circ
[\quad]_{\As_{\bar p}}= [\quad]_{\As_{\bar p}} \circ \Delta^{F_{\bar
      p}}(\zeta)$ (cf. (\ref{leibD'})), we have
\eq
\big[[u,v]_{\As_{\bar p}}\,,\,z\big]_{\As_{\bar p}}=-\big[\,\overline{R}{}^\al_{\bar
  p}(z)\,,\, 
\overline{R}_{\bar p\;\al}\big([u,v]_{\As_{\bar p}}\big)\, \big]_{\As_{\bar
    p}}=
-\big[\,\overline{R}{}^\be_{\bar
  p}\, \overline{R}{}^\al_{\bar
  p}(z)\,,\, \big[\,
\overline{R}_{\bar p\;\be} (u),\overline{R}_{\bar
  p\;\al}(v)\, \big]_{\As_{\bar p}}\, \big]_{\As_{\bar p}}~,
\ee
which implies 
\eq
\big[[u^{\bar\phi},v^{\bar\phi}]\,,\,z^{\bar\phi}\big]^{(3)}_{\bar{p}}=
-\big[\,\overline{R}{}^\be_{\bar
  p}\,\overline{R}{}^\al_{\bar
  p}(z)\,,\, \big[\, 
\overline{R}_{\bar p\;\be} (u),\overline{R}_{\bar
  p\;\al}(v)\, \big]\, \big]^{(3)}_{\bar p}
\ee
where we used the notation 
\eq
u^{\bar\phi}\otimes v^{\bar\phi}\otimes
z^{\bar\phi}:=\e^{-2R}(u\otimes v\otimes z)~.
\ee
 
We can now induce the deformed Jacobi identity for the 3-bracket
$[\quad[\quad]]^{(3)}_{\bar p}$ from the phase space
$\star_\FF$ Jacobi identity of Proposition \ref{FFJacobi}. We let
$(\,\uu\,, \uv\,,\uz\,)=( \pi^*u, \pi^*v, \pi^*z)$ and then pullback along
$s^*_{\bar p}$. Moving $F$ to the right of $F'$, so that its action
becomes trivial on the image of $\pi^*$, the last term simplifies as
\bea
&& s^*_{\bar p}\big( \big[\, \overline{\cal R}{}^\al(\pi^* v^{\pphim})\,,\,
\big[\,\overline{\cal
  R}_\al(\pi^*u^{\pphim}),\pi^*z^{\pphim}\,\big]_{\star_\FF} \,
\big]_{\star_\FF} \big) \nn\\
&&=s^*_{\bar p}\big( [\quad]_{F'}\circ(\Id\ot [\quad]_{F'})\circ (\Id\ot
\Delta^{F'})F^{-1}\, (1\ot 
F^{-1})\, (F'\,^2\ot 1)\, \e^{4R} (\pi^*v\ot \pi^*u\ot \pi^*z)\big)\nn\\[4pt]
&&=s^*_{\bar p}\big( [\quad]_{F'}\circ(\Id\ot [\quad]_{F'})\circ 
(\Delta^{F'}\ot \Id)\, F^{-1}\, (F'\,^2\ot 1)\, \e^{2R} (\pi^*v\ot \pi^*u\ot
\pi^*z)\big)\nn\\[4pt]
&&=s^*_{\bar p}\big( [\quad]_{F'}\circ(\Id\ot [\quad]_{F'})\circ \e^{-R}\,
(F'\,^2\ot 1)\, \e^{2R} (\pi^*v\ot \pi^*u\ot \pi^*z)\big)\nn\\[4pt]
&&= [\quad]_{\As_{\bar p}}\circ(\Id\ot [\quad]_{\As_{\bar p}})\circ \e^{-R}\,
\big(\, \overline{R}_{\bar p}{}^{\al} (v^{\bar\phi})\ot
\overline{R}_{\bar p\;\al} (u^{\bar\phi})\ot z^{\bar\phi}\, \big)\nn\\[4pt]
&&= 
\big[\, \overline{{R}}_{\bar p}{}^\al(v^{\bar \phi})\,,\,\big[\, \overline{{R}}_{\bar
  p\,\al}( u^{\bar \phi}),z^{\bar \phi}\, \big]\,
\big]^{(3)}_{\bar{p}} \ .
\eea
In the second line we used ${\cal R}^{-1}=F'\,^2\, F^2$, recalled (\ref{betauvz}) and
then the antisymmetry of the $R$-flux components $R^{ijk}$. In the third line we used
Corollary~\ref{coroll}, while in the fourth line we used the second identity of Proposition~\ref{pdf} and the equality
$(\Delta\ot\Id)F^{-1}\, (F'\,^2\ot 1)=\e^{-2R}\, (F'\,^2\ot 1)
(\Delta\ot\Id)F^{-1}$ that follows for example from the
Baker-Campbell-Haudorff formula (\ref{BCH}). In the fifth line we used
again antisymmetry of the $R$-flux components and recalled that $R_{\bar p}^{-1}=F_{\bar p}^2$.
We therefore conclude that the deformed Jacobi identity reads as
\eq
\big[u\,,\,[v,z]\big]_{\bar{p}}^{(3)}+ \big[\,\overline{R}{}^\be_{\bar
  p}\, \overline{R}{}^\al_{\bar
  p}(z)\,,\, \big[\, 
\overline{R}_{\bar p\;\be} (u),\overline{R}_{\bar
  p\;\al}(v)\, \big]\, \big]^{(3)}_{\bar p}=
\big[\, \overline{{R}}_{\bar p}{}^\al(v^{\bar \phi})\,,\,\big[\, \overline{{R}}_{\bar
  p\,\al}( u^{\bar\phi}),z^{\bar \phi}\,\big]\, \big]^{(3)}_{\bar{p}}~.
\ee

These deformed Lie algebra expressions simplify considerably on the leaf with momentum
$\bar{p}=0$, which gives the deformed Lie algebra  $\vect(M)_0$ of infinitesimal
diffeomeorphisms on configuration space; it is characterized by an 
undeformed 2-bracket $[u,v]^{(2)}_0=[u,v]$, and the 3-bracket 
\eq 
\big[u\,,\,[v,z]\big]^{(3)}_0=[\quad]\circ(\Id\ot
[\quad])\circ \e^{-R}(u\ot v\ot z)
\ee
that satisfies the deformed Jacobi identity
\eq
\big[u\,,\,[v,z]\big]_{0}^{(3)}+\big[z\,,\,[u,v]\big]^{(3)}_{0}+
\big[v^{\bar \phi}\,,\, [z^{\bar \phi},u^{\bar \phi}] \big]^{(3)}_{0}=0\label{DJI}
~.\ee
We have derived this deformed Lie algebra structure from the one on phase
space, however it is easy to verify (\ref{DJI}) directly.
We can also characterize the deformation via the  Jacobiator 
\eq
[u,v,z]_0:=
\big[u\,,\,[v,z]\big]_{0}^{(3)}+\big[z\,,\,[u,v]\big]^{(3)}_{0}+
\big[v\,,\, [z,u]\big]^{(3)}_{0}=
\big[v\,,\, [z,u]\big]^{(3)}_{0}-\big[v^{\bar \phi}\,,\, [z^{\bar
  \phi},u^{\bar \phi}] \big]^{(3)}_{0}~,
\ee
i.e., $[\quad]_0=[\quad]\circ(\Id\ot
[\quad])\circ \sigma_{12}\circ (\e^{-R}-1)$. We have thus derived a general
expression to all orders in the $R$-flux components $R^{ijk}$ for the Jacobiator on
arbitrary vector fields; this formula nicely illustrates how nonassociativity
can be explicitly manifest  in a theory of gravity on configuration
space.

Finally we study Lie derivatives of forms on configuration space. We define the Lie
  derivative $\Lie^{\bar p}$ by
\eq
\Lie^{\bar p}{\!}_u(\eta):=s^*_{\bar p}\circ
\Lie^{\star_\FF}_{\pi^*u}\big(\pi^*\eta\big)~,
\ee
for all $u\in \vect(M)_{\bar p}$,
  $\eta\in \Omega^\bullet(M)_{\bar p}$.  As in (\ref{lievect}), since no partial derivatives in the momentum
directions appear, and since $s^*_{\bar p\,} \dd p_i=0$, we find explicitly
\eq
\Lie^{\bar p}{\!}_u(\eta)=\Lie_{\overline{F}_{{\bar p}}^\al(u)}\big(\,
{\overline{F}_{{\bar p}\,\al}(\eta)}\, \big)=:\Lie^{\As_{\bar  p}}_u(\eta) 
\ee
where we used the notation introduced in (\ref{eq469}), and
$\Lie^{\As_{\bar  p}}$ is the ${\As_{\bar  p}}$ Lie derivative on
Moyal--Weyl noncommutative space (cf. \cite{GR2,book}).
In order to write the Leibniz rule for this infinitesimal
diffeomorphism we observe from (\ref{eq443}) that the exterior product
$\wedge^{(2)}_{\bar p}$ is the usual exterior product in Moyal--Weyl
noncommutative space: $\wedge^{(2)}_{\bar p}=\wedge_{\As_{\bar p}}$;
since the Lie derivative is also the same, we have
\eq\label{LeiblieCS}
\Lie^{\bar p}{\!}_u(\eta\wedge_{\As_{\bar p}}
\omega)=\Lie^{\bar p}{\!}_{u}(\eta)\wedge_{\As_{\bar p}}\omega+
\overline{R}_{\bar p}{}^\al(\eta)\wedge_{\As_{\bar p}}
\Lie^{\bar p}_{~\;\overline{R}_{\bar p\,\al}(u)}(\omega)~.
\ee 
In particular if we pullback this expression to the leaf
$\bar{p}=0$ we obtain the usual undeformed Lie derivative action.

\subsection*{Acknowledgments}

We thank R.~Blumenhagen, L.~Castellani, B.~Jur\v{c}o, D.~L\"ust,
D.~Minic, D.~Mylonas, A.~Schenkel and
P.~Schupp for helpful discussions, and G.~Amelino-Camelia, P.~Martinetti
and J.-C.~Wallet for the invitation to submit this
contribution to their proceedings. This work was initiated while
R.J.S. was visiting the Department of Theoretical Physics of the
University of Torino during January--February 2014, and completed
while he was visiting the Alfred Renyi Institute of Mathematics
in Budapest during March--April 2015; he warmly thanks both
institutions for support and hospitality during his
stays there. The work of R.J.S. was supported in
part by the Consolidated Grant ST/L000334/1 from the UK Science and
Technology Facilities Council.

\end{document}